\documentclass[journal,comsoc]{IEEEtran}
\usepackage{amsmath,amsfonts,amssymb}
\usepackage{amsthm}
\usepackage{enumerate}
\usepackage{algorithmic}
\usepackage{algorithm}
\usepackage{array}
\usepackage{textcomp}
\usepackage{stfloats}
\usepackage{url}
\usepackage{verbatim}
\usepackage{graphicx}
\usepackage{cite}
\usepackage{multirow}
\usepackage{makecell}
\usepackage{xcolor}
\usepackage{diagbox}
\usepackage{tabularx}
\newtheorem{theorem}{Theorem}
\newtheorem{lemma}{Lemma}
\newtheorem{corollary}{Corollary}

\newcolumntype{Y}{>{\centering\arraybackslash}X}

\DeclareMathOperator*{\argmin}{argmin\,}
\newcommand\norm[1]{\left\lVert#1\right\rVert}
\newcommand{\revision}[1]{{#1}}

\begin{document}

\title{MIMO Channel Estimation using Score-Based Generative Models}

\author{Marius Arvinte, \IEEEmembership{Member, IEEE}, and Jonathan I. Tamir, \IEEEmembership{Member, IEEE}
        % <-this % stops a space
\thanks{M. Arvinte was with the Electrical and Computer Engineering Department, University of Texas at Austin, TX, USA, and is now with Intel Corporation. (arvinte@utexas.edu). J. I. Tamir is with the Chandra Family Department of Electrical and Computer Engineering Department, the Department of Diagnostic Medicine, and the Oden
Institute, University of Texas at Austin, TX, USA (jtamir@utexas.edu). This work was supported by ONR grant N00014-19-1-2590, NSF IFML 2019844, and a gift made by InterDigital, an affiliate of the 6G@UT center within the Wireless Networking and Communications Group at The University of Texas at Austin. A portion of the ideas in this manuscript have appeared as a workshop paper at WCNC 2022. \copyright2022 IEEE. Personal use of this material is permitted. Permission
from IEEE must be obtained for all other uses, in any current or future media, including reprinting/republishing this material for advertising or promotional purposes, creating new collective works, for resale or redistribution to servers or lists, or reuse of any copyrighted component of this work in other works.}% <-this % stops a space
}

\markboth{To Appear in IEEE Transactions on Wireless Communications}%
% \markboth{2022}%
{Shell \MakeLowercase{\textit{et al.}}: A Sample Article Using IEEEtran.cls for IEEE Journals}

\maketitle
\begin{abstract}
Channel estimation is a critical task in multiple-input multiple-output (MIMO) digital communications that substantially effects end-to-end system performance. In this work, we introduce a novel approach for channel estimation using deep score-based generative models. A model is trained to estimate the gradient of the logarithm of a distribution and is used to iteratively refine estimates given measurements of a signal. We introduce a framework for training score-based generative models for wireless MIMO channels and performing channel estimation based on posterior sampling at test time. We derive theoretical robustness guarantees for channel estimation with posterior sampling in single-input single-output scenarios, and experimentally verify performance in the MIMO setting. Our results in simulated channels show competitive in-distribution performance, and robust out-of-distribution performance, with gains of up to $5$ dB in end-to-end coded communication performance \revision{compared to supervised deep learning methods}. \revision{Simulations} on the number of pilots show that high fidelity channel estimation with $25$\% pilot density is possible for MIMO channel sizes of up to $64 \times 256$. Complexity analysis reveals that model size can efficiently trade performance for estimation latency, and that the proposed approach is competitive with compressed sensing in terms of floating-point operation (FLOP) count.
\end{abstract}


\section{Introduction}
\IEEEPARstart{M}{assive} MIMO represents a key technology in fifth generation (5G) and envisioned sixth generation (6G) communication systems \cite{viswanathan2020communications,heng2021six}, promising to increase communication reliability by orders of magnitude without increasing bandwidth requirements. With the deployment of millimeter wave (mmWave) band communications, recovering accurate, high-dimensional channel state information (CSI) using reduced pilot overhead has become a major open research problem \cite{he2018deep,balevi2020high,hong2021role}. For example, reflective intelligent surfaces -- equipped with up to hundreds of antennas -- are an active area of research for high-dimensional wireless systems \cite{shlezinger2021dynamic,zheng2022survey}, and channel estimation has been recently investigated using optimization approaches \cite{jensen2020optimal}, and data-driven methods \cite{zheng2019intelligent,liang2021reconfigurable}. Estimating accurate CSI with data-driven methods is an important area of research for future communication systems that integrate artificial intelligence in physical layer processing \cite{letaief2019roadmap,viswanathan2020communications}, including being taken into consideration for standardization \cite{han2020artificial}.

An important challenge for estimation algorithms is their robustness to test-time distributional shifts \cite{koh2021wilds} that naturally occur when the test environment no longer matches the algorithm design conditions. This is present in wireless communication systems, especially at the user side, where the propagation conditions may change from indoor to outdoor \cite{palacios2017tracking}, whenever the user is moving. An important question is whether deep learning-aided channel estimation can retain its performance, both from theoretical, and practical perspectives \cite{hu2020deep}. To this end, the main motivation of this work is to develop robust, data-driven, deep learning based MIMO channel estimation algorithms for high-dimensional communication scenarios.

In this paper, inspired by recent results that show the potential of score-based (diffusion) generative models for specialized applications such as magnetic resonance imaging (MRI) reconstruction \cite{jalal2021robust,song2022solving}, we introduce a training and inference algorithm for wireless channel estimation using score-based generative models in a single-carrier, point-to-point MIMO communication scenario. We use a distribution learning approach for modeling high-dimensional, mmWave MIMO channels in a stochastic environment. We model the log-distribution of channels by learning the high-dimensional gradient -- known as the \textit{score}. To learn the score of the distribution, we use score-based generative models, originally introduced in \cite{song2019generative}, and that have so far been primarily demonstrated on natural image benchmark datasets \cite{dhariwal2021diffusion}. For training, a database of known channels is used to train a score-based generative model in an unsupervised manner, that is independent of the pilot symbols. During test time, we perform probabilistic channel estimation by sampling from the posterior distribution \revision{using annealed Langevin dynamics}. We tackle the challenges that arise when performing channel estimation in an out-of-distribution setting, \revision{i.e.,} to environments not seen during training, as well as in a very wide (up to $40$ dB) signal-to-noise ratio (SNR) range \revision{and in interference scenarios}. \revision{Numerical simulations} show that using the proposed approach, estimation performance is retained even when the simulated test channels come from a different distribution than the one during training, and gains of up to $5$ dB in energy-per-bit to noise ratio ($E_b/N_0$) are achieved against competing deep learning methods in coded end-to-end communication system simulations. We also investigate the computational complexity of the proposed approach, and identify promising future research directions in developing more compact and efficient deep neural networks to reduce estimation latency.

\subsection{Related Work}
Modern statistical wireless channel modeling stems from the Saleh-Valenzuela clustered channel model \cite{saleh1987statistical}, which characterizes a propagation environment through a sum of impulse responses with stochastic delays and amplitudes. Extensions for MIMO channels have added angles of departure and arrival on both ends of the communication link \cite{spencer2000modeling}, and have included beamforming effects in the modeling \cite{rappaport2012broadband}. Standardized clustered delay line (CDL) models for cellular communications are extremely flexible and can model wireless channels in four dimensions (time, frequency, transmit antennas, and receive antennas), and are adopted in 5G specifications \cite{3gpp.38.901}. In this work, we use the well-established CDL family of stochastic environments to evaluate the performance of MIMO channel estimation algorithms. 
Beyond their usage for synthetic channel generation and performance evaluation, statistical channel models have been used to aid compressed sensing (CS) for MIMO channel estimation, based on an assumed low rank channel model in the beamspace domain \cite{venugopal2017channel}. The Lasso is a fast estimation method that imposes sparsity in the two-dimensional Fourier (2D-DFT) domain through $\ell_1$-regularization in the channel estimation objective. \revision{The EM-GM-AMP algorithm \cite{schniter2014channel} uses the approximate message passing (AMP) algorithm as a backbone in recovering a Gaussian mixture of impulses in the beamspace domain, from a number of undersampled and noisy measurements. Along the same line,} the atomic norm decomposition \cite{bhaskar2013atomic} and its extension to fast approximate atomic norm decomposition (fsAD) \cite{zhang2017atomic} use knowledge of the antenna array shape and the underlying structure of the CDL models (sparsity in the continuous 2D angular domain) to achieve competitive channel estimation results. In contrast to CS methods, score-based models do not assume a specific channel model, and thus are likely amenable to use in real-world wireless propagation environments where sparsity and low rank assumptions may not always hold \cite{molisch2016millimeter}.

The work in \cite{balevi2020high} is an application of the compressed sensing with generative models (CSGM) framework \cite{bora2017compressed} to wireless channel estimation. This method trains a deep generative adversarial network (GAN) with a low-dimensional latent space using a training set of channels, and at test-time formulates an optimization problem to recover the channel state information matrix using the received pilot symbols and the pre-trained model. In contrast to GANs, score-based generative models do not use adversarial training, and make less restrictive assumptions about the low-dimensional nature of the wireless channels, instead aiming to learn the score of the distribution even in regions with a low probability density. \revision{The recent work in \cite{yang2022learning} introduces an unsupervised learning approach that jointly learns the best hyper-parameters and a refinement model for optimal precoding matrices. This is similar to the iterative denoising approach done with Langevin dynamics.} Conditional generative models have been used to learn the effects of unknown channels \cite{o2019approximating}, where a deep model is trained to map transmitted symbols to received symbols. The goal of this approach is to model transmitter and receiver effects (e.g., nonlinear power amplification or analog-to-digital conversion) alongside channel effects, which is useful for learning black-box approximations to entire communication chains \cite{ye2020deep}.

Supervised end-to-end training of deep learning based methods has been successfully used for wireless MIMO channel estimation. The algorithm proposed in \cite{soltani2019deep} introduces a two-stage deep learning approach for two-dimensional channel estimation: in the first stage, a super-resolution network is used, followed by a denoising stage. While originally introduced for time-frequency channels, this approach is applicable to MIMO channels as well, where an initial estimate of the channel matrix is treated as an image \cite{huang2018deep}. The work in \cite{he2018deep} introduces a powerful and robust deep learning algorithm in the form of the learned denoising approximate message passing (L-DAMP) algorithm \cite{metzler2017learned}. In L-DAMP, differentiable optimization steps are interleaved with forward passes of a deep neural network, and the entire method is trained end-to-end using back-propagation.

Score-based generative models -- closely related to diffusion models \cite{song2021maximum} -- are introduced in \cite{song2019generative}, where the authors also introduce annealed Langevin dynamics and how to efficiently sample from a distribution of interest using a learned score model. These models have been recognized to surpass generative adversarial networks (GANs) at modeling real-world distributions, such as natural images \cite{dhariwal2021diffusion}. The work in \cite{song2020improved} introduces practical techniques for improving the training of score-based models, which we leverage and adapt for wireless channels. The work in \cite{jalal2021instance} theoretically proves the in-distribution optimality of posterior sampling and links this to the full-dimensional support coverage capability of a model.

\subsection{Contributions}
Our contributions in this work are the following:
\begin{itemize}
    \item We propose a posterior sampling based solution for MIMO channel estimation. Given received pilots, at test (inference) time the algorithm samples from the posterior distribution of channel state information conditioned on the received pilots. In the training phase, a denoising score matching framework learns the score (the gradient of log-prior distribution) of high-dimensional MIMO channels. During inference, we use the learned score-based model in conjunction with the received pilots to iteratively update the channel estimate and perform posterior sampling (Algorithm~\ref{alg:inference}).
    \item We derive a closed-form expression that describes robustness guarantees for multi-tap, single-input single-output (SISO) channel estimation with posterior sampling in an out-of-distribution setting, under assumptions of mutual independence between the tap gains and locations. The derived expression bounds the probability of successful recovery using the \textit{mismatch-to-noise} ratio. This predicts the behaviour of the proposed method as a function of the SNR and the distributional mismatch between training and testing environments.
    \item We perform \revision{numerical simulations} showing the performance of score-based models and a set of diverse baselines on the CDL family of channel models \cite{3gpp.38.901} at mmWave frequencies. Our main findings are that score-based generative models are near-optimal in-distribution, \revision{and are suitable in scenarios where high-fidelity channel estimation is required and an increased computational complexity is acceptable, such as mmWave-based fixed access or backhaul}. We evaluate the complexity of the proposed approach, and find that it favourably scales in terms of FLOP count and estimation latency when the system size increases, compared to CS-based approaches. To facilitate reproducible deep learning research, we release open-source code for the experiments\footnote{\url{https://github.com/utcsilab/score-based-channels}}.
\end{itemize}

\subsection{Notation and Organization}
We use boldfaced letters $\mathbf{x}, \mathbf{X}$ to denote vectors and matrices, respectively, and $x_{i,j}$ to denote the $(i,j)$-th element of a matrix. We use italic letters to denote their random variable counterparts. The notation $p_X( \mathbf{X})$ represents the continuous probability distribution function of the matrix random variable $X$, evaluated at the matrix $\mathbf{X}$. We use $\mathbf{X}^\mathrm{H}$ to denote the Hermitian (conjugate transpose) of $\mathbf{X}$. The matrix $\mathbf{I}$ denotes the identity matrix of appropriate size. We use $\norm{\cdot}_F$ to denote the Frobenius norm of a matrix, $\norm{\cdot}_{1,1}$ to denote the $\ell_1$-norm of a vectorized matrix, $\mathcal{O}$ to denote proportionality up to a multiplicative constant, and $\log$ to denote the natural logarithm. We use $\nabla f(\mathbf{X})$ to denote the derivative of a function $f$ with respect to a value $\mathbf{X}$, as the matrix with the $(i,j)$-th entry given by $\partial f( \mathbf{X}) / \partial x_{i, j}$. \revision{We use $dX$ to denote an infinitesimal change in a random process $X$.}

The remainder of the paper is structured as follows: Section~\ref{sec:system} introduces the MIMO system model and defines the score of a distribution. Section~\ref{sec:proposed} introduces the training and inference (testing) stages of the proposed approach, while Section~\ref{sec:theory} introduces our theoretical derivations that characterize the probability of successful estimation. Section~\ref{sec:experiments} presents \revision{numerical simulation results} and discussions, and Section~\ref{sec:conclusion} concludes the paper.

\section{Preliminaries}
\label{sec:system}
\subsection{Wireless System Model}
\label{subsec:system_model}
We consider a narrowband, point-to-point MIMO communication scenario between a transmitter and receiver with $N_\mathrm{t}$ and $N_\mathrm{r}$ antennas, respectively. Propagation in this scenario is characterized by the channel state information matrix $\mathbf{H} \in \mathbb{C}^{N_\mathrm{r}\times N_\mathrm{t}}$. We let $\mathbf{p}_i \in \mathbb{C}^{N_\mathrm{t}}$ denote the $i$-th pilot symbol chosen from a pre-designed codebook of $N_\mathrm{p}$ pilots, and $\mathbf{W} \in \mathbb{C}^{N_\mathrm{r} \times N_\mathrm{r}}$ be the receive beamforming matrix. The received signal vector for the $i$-th pilot $\mathbf{y}_i$ is given by:
\begin{equation}
    \mathbf{y}_i = \mathbf{W}^\mathrm{H}(\mathbf{H} \mathbf{p}_i s + \mathbf{n}_i),
    \label{eq:basic_model}
\end{equation}
\noindent where $\mathbf{n}_i$ is complex additive white Gaussian noise with zero mean and power covariance matrix $\sigma_\mathrm{pilot}^2 \mathbf{I}$, and $s$ is a complex-valued scalar. In practice, the pilot vectors are selected as entries from a beamforming codebook with structural constraints \cite{heath2016overview}. For the remainder of the paper, we make no assumptions on the structure of $\mathbf{p}_i$. We assume that $s=1$, $\mathbf{W} = \mathbf{I}$ (fully digital receiver), and that $\mathbf{p}_i$ is constrained to have unit amplitude and low-resolution phase -- each entry of $\mathbf{P}$ is a randomly chosen (fixed for all test samples) quadrature phase shift keying (QPSK) symbol. Assuming that the channel state information is constant across $N_\mathrm{p}$ pilot transmissions, we obtain the matrix model:
\begin{equation}
    \mathbf{Y} = \mathbf{H}\mathbf{P} + \mathbf{N}.
    \label{eq:matrix_model}
\end{equation}
\noindent Channel estimation requires estimating the channel state information matrix $\mathbf{H}$ using the received pilot matrix $\mathbf{Y}$, while having knowledge of the transmitted pilot matrix $\mathbf{P}$. The latter part is common in communication standards, where pilot sequences are pre-specified \cite{heath2016overview}. Let $\alpha = N_\mathrm{p} / N_\mathrm{t}$ denote the pilot density. When $\alpha < 1$, then there are $N_\mathrm{r} N_\mathrm{p} < N_\mathrm{r} N_\mathrm{t}$ received pilots, and channel estimation is an under-determined inverse problem.

\subsection{\revision{Langevin Dynamics}}
\label{subsec:langevin_dynamics}
\revision{Langevin dynamics are a process for sampling from arbitrary distributions. Let $p_H$ be an arbitrary, high-dimensional probability distribution and $H_t$ be a Langevin diffusion process on the same probability space as $p_H$, defined by the dynamics \cite{roberts1996exponential}:}
\begin{equation}
    \revision{\mathrm{d} H_t = \mathrm{d} W_t + \frac{1}{2} \nabla \log p_H ( H_t ),}
\label{eq:continuous_langevin}
\end{equation}
\noindent \revision{where $W_t$ is an i.i.d. zero mean and unit variance Gaussian random process. A remarkable property of the Langevin diffusion process is that the stationary distribution of $H_t$ converges in probability to $p_H$ as $t \rightarrow \infty$, regardless of the initial value $\mathbf{H}_0$, assuming that certain smoothness conditions for $p_H$ are met \cite{roberts1996exponential}.} \revision{Running the process in \eqref{eq:continuous_langevin} for a given initial value is thus a way to sample from the distribution $p_H$. In practice, numerical machine computations require finite precision and updates to $\mathbf{H}$. The discretized Langevin dynamics update equation with step size $\epsilon$ is given by \cite{roberts1996exponential}:}
\begin{equation}
    \revision{\mathbf{H}_{t+1} \leftarrow \mathbf{H}_t + \epsilon \cdot \nabla \log p_H ( \mathbf{H}_t ) + \sqrt{2\epsilon} \cdot \zeta_t },
\label{eq:discrete_langevin}
\end{equation}
\noindent \revision{where $\zeta_t$ is i.i.d. Gaussian noise with zero mean and unit variance.} \revision{The work in \cite{roberts1996exponential} analyzes the convergence rate of \eqref{eq:discrete_langevin} and finds that convergence (also termed \textit{mixing}) can be slower than the continuous process in \eqref{eq:continuous_langevin}. To alleviate slow mixing times, practical modifications have been proposed to \eqref{eq:discrete_langevin}. In this work, we leverage \textit{annealed} Langevin dynamics \cite{song2019generative}, where time-varying hyper-parameters $\alpha_t$ and $\beta_t$ are introduced to improve convergence and yield the update rule:}
\begin{equation}
    \revision{\mathbf{H}_{t+1} \leftarrow \mathbf{H}_t + \alpha_t \cdot \nabla \log p_H ( \mathbf{H}_t ) + \beta_t \cdot \zeta_t.}
    \label{eq:prior_sampling}
\end{equation}
\noindent \revision{The above consists of two additive update terms to the current iterate $\mathbf{H}_t$:}
\begin{itemize}
    \item \revision{$\alpha_t \cdot \nabla \log p_H ( \mathbf{H}_t )$ increases the likelihood of the current sample. Because the gradient is the local direction of steepest ascent, this update term guides $\mathbf{H}_t$ to a more plausible sample under the distribution $p_H$, i.e., towards a point of the distribution with higher density.}
    \item \revision{$\beta_t \cdot \zeta_t$ represents a perturbation to the above process. If the hyper-parameter $\beta_t$ is chosen correctly, this perturbation allows for sample diversity and prevents always sampling the mode of the distribution \cite{jalal2021instance}.}
\end{itemize}
\revision{Given a target criterion and a set of validation samples, these hyper-parameters can be tuned using an approximate grid search approach. The only issue that remains is how can one obtain query access to $\nabla \log p_H ( \mathbf{H} )$, for arbitrary $\mathbf{H}$. This is the central purpose of score-based generative modeling and described next.}

\subsection{The Score of a Distribution}
An important assumption of this paper is that wireless channels are sampled from a distribution. This is a realistic assumption for real-world wireless environments, that has been extensively used in both theoretical \cite{weichselberger2006stochastic} and practical works \cite{o2019approximating}. Let $p_H$ denote the distribution of complex-valued MIMO channels for an arbitrary, stochastic environment introduced in Section~\ref{subsec:system_model}. The \textit{score} of $p_H$ at $\mathbf{H}$ is defined as \cite{vincent2011connection}:
\begin{equation}
    \psi_H( \mathbf{H} ) =  \nabla \log p_H(\mathbf{H}),
\end{equation}
\noindent where $\psi_H ( \mathbf{H} ) \in \mathbb{C}^{N_\mathrm{r} \times N_\mathrm{t}}$. 
While the score function can be used to sample from $p_H$ using annealed Langevin dynamics, it is generally intractable for non-trivial distributions, including realistic models of wireless channels. The result in \cite{vincent2011connection} shows that it is possible to learn an approximation to the score function using an unsupervised formulation. For the remainder of this work, we do not make any explicit assumptions on $p_H$ or $\psi_H$. That is, we do not assume the existence of a low-dimensional (sparse) representation of MIMO channels, and the proposed recovery approach does not use explicit information about $p_H$ (the channel statistics in a given environment) or $\psi_H$, instead fully relying on a data-driven approach.

\subsection{Score-Based Generative Modeling}
\label{subsec:score_models}
The goal of score-based generative models is to learn the score function $\psi_H$ at all input points, for a channel distribution $p_H$. When training samples $\{\mathbf{H}_i\}_{i=1}^N$ and their corresponding scores $\{\psi_H ( \mathbf{H}_i )\}_{i=1}^N$ are available, explicit score matching uses the following loss function to train a model $s_\theta$ \cite{vincent2011connection}:
\begin{equation}
    \mathcal{L}_{\mathrm{ESM}, p_H} ( \theta ) = \mathbb{E}_{\mathbf{H} \sim p_H} \left[ \norm{s_\theta ( \mathbf{H} ) - \psi_H ( \mathbf{H} ) }_2^2 \right].
\label{eq:esm}
\end{equation}
The above can be minimized using any gradient descent method. However, in general, $\psi_H$ is intractable, and thus the loss objective in \eqref{eq:esm} cannot be used to learn a model $s_\theta$. Furthermore, if $\psi_H$ were tractable and computationally feasible to evaluate, a deep learning model that approximates it would have limited practical use. To circumvent this problem, the work in \cite{vincent2011connection} proposes to use \textit{denoising} score matching by synthesizing corrupted data samples $\tilde{\mathbf{H}}$ and learning the score of the conditional distribution $p_{\tilde{H} | H}$, using the following objective:
\begin{equation}
\label{eq:dsm}
    \mathcal{L}_\mathrm{DSM} ( \theta ) =
    \mathbb{E}_{\mathbf{H} \sim p_H, \tilde{\mathbf{H}} \sim p_{\tilde{H}} } \left[ \norm{s_\theta ( \tilde{\mathbf{H}} ) - \nabla \log p_{\tilde{H} | H} ( \tilde{\mathbf{H}} | \mathbf{H} ) }_2^2 \right].
\end{equation}
The work in \cite{vincent2011connection} proves the following theorem:
\begin{theorem}[Appendix from \cite{vincent2011connection}]
\label{thm:vincent}
Assuming that $\log p_{\tilde{H} | H} ( \tilde{\mathbf{H}} | \mathbf{H} )$ is differentiable with respect to $\tilde{\mathbf{H}}$, then the losses $\mathcal{L}_{\mathrm{ESM}, p_{\tilde{H}}}$ and $\mathcal{L}_\mathrm{DSM}$ are equivalent.
\end{theorem}
Theorem~\ref{thm:vincent} implies that the score of the perturbed distribution $p_{\tilde{H}}$ can be learned by using the objective in \eqref{eq:dsm}, as long as its probability distribution is differentiable. \revision{In particular, this allows the use of arbitrary noise distributions for training, and learning the score at arbitrarily perturbed inputs, when using a continuum of noise levels.} When the perturbation $\mathbf{Z}$ is i.i.d. Gaussian, with zero mean and covariance matrix $\sigma_z^2 \mathbf{I}$, we obtain that \cite{vincent2011connection}: 
\begin{equation}
    \nabla \log p_{\tilde{H} | H} ( \tilde{\mathbf{H}} | \mathbf{H} ) = -  \mathbf{Z} / \sigma_z^2.
\end{equation}
The work in \cite{song2019generative} proposes to train a single score-based model using a weighted version of the loss in \eqref{eq:dsm} at multiple noise levels, as well as a learnable model (in practice, a deep neural network) $s_\theta$ with parameters (weights) $\theta$. The loss function for this model is given by \cite{song2019generative}:
\begin{equation}
\label{eq:weighted_score_loss}
    \mathcal{L}_\mathrm{score} ( \theta ) = \mathbb{E}_{j, \mathbf{H} \sim p_H, \mathbf{Z}_j \sim p_{Z_j}} \left[ \sigma_{z_j}^2 \norm{s_\theta (\mathbf{H} + \mathbf{Z}_j ) + \frac{\mathbf{Z}_j}{\sigma_{z_j}^2}}_2^2 \right].
\end{equation}
The decision to weigh the predicted score at each noise level comes from formulating denoising score-matching as a variance-exploding (VE) diffusion process \cite{song2021maximum}. As $-\mathbf{Z} / \sigma^2$ tends towards infinity magnitude for small $\sigma$, weighting with $\sigma^2$ compensates for this and stabilizes learning. \revision{Finally, the learned score function $\psi_{\tilde{H} | H}$ is used to draw samples from the posterior distribution in the sequel.}

\subsection{Posterior Sampling Using Score Functions}
The formulation in Section~\ref{subsec:score_models} does not require any assumptions on $\mathbf{P}$, $\mathbf{Y}$ or $\sigma_\mathrm{pilot}$, and makes learning the score function an unsupervised task. To perform channel estimation with a learned model, we use \textit{posterior sampling} via annealed Langevin dynamics, with the update given by:
\begin{equation}
\label{eq:posterior_update}
    \mathbf{H} \leftarrow \mathbf{H} + \alpha \cdot \psi_{H|Y} ( \mathbf{H} | \mathbf{Y}) + \beta \cdot \zeta,
\end{equation}
\noindent where $\alpha$ and $\beta$ are decaying step sizes, potentially different than the ones in \eqref{eq:prior_sampling}. \revision{Using Bayes rule for $p_{H|Y}(\mathbf{H} | \mathbf{Y}) = \frac{p_{Y|H}(\mathbf{Y} | \mathbf{H}) \cdot p_{H} (\mathbf{H}) }{p_{Y} (\mathbf{Y})}$ and expanding the logarithm yields}:
\begin{equation}
    \revision{\log p_{H | Y} (\mathbf{H} | \mathbf{Y}) = \log p_{Y | H}  (\mathbf{Y} | \mathbf{H}) + \log p_{H} (\mathbf{H}) - \log p_{Y} (\mathbf{Y}).}
\end{equation}
\revision{Taking the gradient with respect to $\mathbf{H}$ on both sides, we have that $\nabla \log p_{Y} (\mathbf{Y}) = 0$ for all $\mathbf{Y}$, and $\psi_{H|Y} (\mathbf{H} | \mathbf{Y}) = \psi_{Y|H} (\mathbf{Y} | \mathbf{H}) + \psi_{H} (\mathbf{H})$. Replacing $\psi_{H|Y} (\mathbf{H} | \mathbf{Y})$ in \eqref{eq:posterior_update} yields}:
\begin{equation}
    \mathbf{H} \leftarrow \mathbf{H} + \alpha \cdot \psi_{Y|H} ( \mathbf{Y} | \mathbf{H} ) + \alpha \cdot \psi_H ( \mathbf{H} ) + \beta \cdot \zeta.
    \label{eq:posterior_sampling}
\end{equation}
Compared to \eqref{eq:prior_sampling}, the additional term $\alpha \cdot \psi_{Y|H}$ updates the current estimate of $\mathbf{H}$ in a direction where it becomes more consistent with the received pilots $\mathbf{Y}$. Taken together, the three updates in \eqref{eq:posterior_sampling} represent the core routine of the proposed channel estimation algorithm.

\section{Methods}
\label{sec:proposed}
There are two optimization problems that must be solved in order to estimate channels using score-based generative models. These are mapped to training and inference stages, respectively. \revision{The decoupling of these two stages is a key difference compared to supervised channel estimation methods, where information about pilots is assumed during training:}
\begin{enumerate}
    \item During the training stage, a score-based generative model $s_\theta$ is trained by minimizing the loss function described in Section~\ref{sec:training}. This stage only takes place once in the lifetime of a wireless device, e.g., offline, using a powerful computational server and a dataset of accurate channel measurements or simulated channel realizations.
    \item During the testing (inference) stage, channel estimation is formulated as an optimization problem and the iterative algorithm in Section~\ref{sec:inference} is used to solve it. This stage uses the pretrained score-based model in conjunction with the received pilots to recover CSI. The formulation in Algorithm~\ref{alg:sampling} is independent from the first stage, \revision{and could accommodate other impairments such as \revision{interference scenarios} or few-bit quantization of the received pilots}.
\end{enumerate}

\subsection{Training a Score-Based Model for Wireless Channels}
\label{sec:training}
\begin{figure}[!t]
\centering
\includegraphics[width=0.95\linewidth]{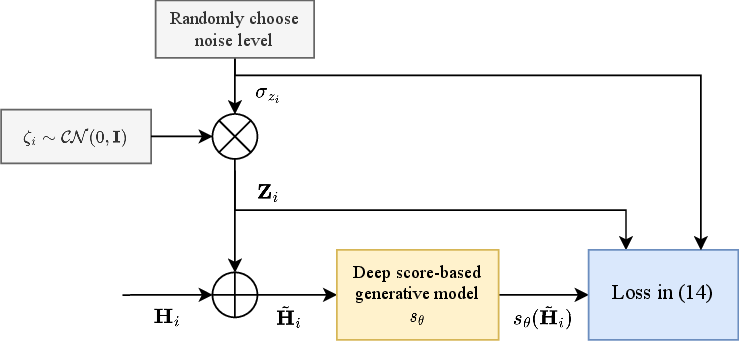}
\caption{Illustration of training flow for a single sample in a batch.}
\label{fig:score_training}
\end{figure}
In practice, a deep neural network with weights $\theta$ is used to learn a model $s_\theta$. \revision{We use the RefineNet architecture \cite{lin2017refinenet}, that operates at multiple resolution levels in parallel, and consists of a series of RefineNet blocks, with residual skips between them and a number of down- and up-sampling operations, respectively. This allows the model to learn structural relations in the CSI at multiple resolution levels, and to efficiently predict the score function. Furthermore, as the model is fully convolutional, it can handle channel matrices of dynamic input size, both during training and testing. The architecture of a RefineNet block is shown in detail in Figure~\ref{fig:score_architecture}.} Further details about the architectural details of RefineNet are available in the original paper \cite{lin2017refinenet}, as well as in our source code repository.
The loss used to train $s_\theta$ with a batch size of $B$ samples is given by the finite-sample version of \eqref{eq:weighted_score_loss} as:
\begin{equation}
    \mathcal{L}_\mathrm{train} ( \theta ) = \frac{1}{B} \sum_{i=1}^B \sigma_{z_i}^2 \norm{s_\theta ( \mathbf{H}_i + \mathbf{Z}_i ) + \frac{\mathbf{Z}_i}{\sigma_{z_i}^2}}_2^2,
\label{eq:batch_loss}
\end{equation}
\noindent where, at each training step, and for each sample in the batch we sample $\mathbf{Z}_i \sim \mathcal{CN} ( 0, \sigma_{z_i}^2 \mathbf{I} )$, where $\sigma_{z_i}$ is a noise level selected uniformly at random from a set of values $\{\sigma_{z_l}\}_{l = 1, \dots, L}$. That is, for each sample in a batch, noise at a randomly chosen noise level is added to the clean CSI matrix and the model is trained to predict the scaled negative noise that points away from the noisy sample $\tilde{\mathbf{H}}_i$ and towards the clean sample $\mathbf{H}_i$ -- this represents the score of the distribution of perturbed wireless channels.
\begin{figure}[!t]
\centering
\includegraphics[width=0.9\linewidth]{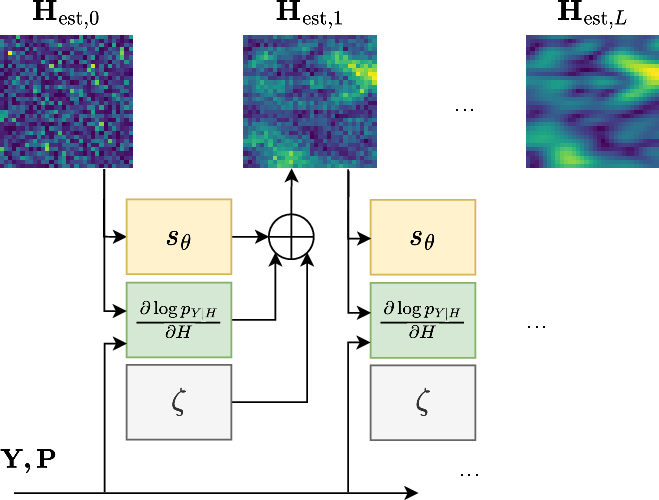}
\caption{Iterative inference procedure for channel estimation using a score-based model together with $\mathbf{Y}$ and $\mathbf{P}$.}
\label{fig:score_inference}
\end{figure}
Figure~\ref{fig:score_training} illustrates this effect for a perturbed sample, starting from the ideal channel matrix $\mathbf{H}_i$, and in which $\tilde{\mathbf{H}}_i$ is used as input to the deep neural network trained to predict $-\mathbf{Z}_i/\sigma^2_{z_i}$. Finally, note that \eqref{eq:batch_loss} does not use any information about \eqref{eq:matrix_model}, the power $\sigma_\mathrm{pilot}^2$ of the noise affecting the received pilots, or the pilot matrix $\mathbf{P}$ itself. \revision{The majority of existing deep learning approaches for MIMO channel estimation are supervised and use this information explicitly during training, which leads to overfitting to a specific measurement or noise distribution. In contrast, score-based generative modeling is unsupervised and does not require this information for training, making inference robust and usable across a wide range of SNR values and number of pilot symbols.}

\subsection{MIMO Channel Estimation via Posterior Sampling}
\label{sec:inference}
To perform channel estimation using score-based models, we resort to \textit{posterior sampling}: our solution comes in the form of a single sample from the posterior distribution $p_{H|Y} ( \cdot | \mathbf{Y} )$. To sample from the posterior, we use annealed Langevin dynamics \cite{song2019generative}. This is an iterative algorithm, that takes the following form at the $i$-th step:
\begin{equation}
    \mathbf{H}_{\mathrm{est},i+1} = \mathbf{H}_{\mathrm{est},i} + \alpha_i \cdot \nabla \log p_{H | Y} ( \mathbf{H}_{\mathrm{est},i} | \mathbf{Y} ) + \sqrt{2\beta \cdot \alpha_i} \cdot \sigma_{z_i} \cdot \mathbf{\zeta},
\end{equation}
\noindent where $\mathbf{\zeta} \sim \mathcal{CN} (0, \mathbf{I} )$ is randomly sampled at every update step, and $\alpha_i = \alpha_0 \cdot r^i$. The scalars $\alpha_0$, $\beta$ and $r$ are hyper-parameters discussed in Section~\ref{sec:experiments}. Using the expansion in \eqref{eq:posterior_sampling} yields:
\begin{align}
\begin{split}
\label{eq:abstract_update}
    \mathbf{H}_{\mathrm{est},i+1} = \mathbf{H}_{\mathrm{est},i} & + \alpha_i \cdot  \Big( \nabla \log p_{Y | H} ( \mathbf{Y} | \mathbf{H}_{\mathrm{est}, i} ) \\ & \qquad \quad + \nabla \log p_H ( \mathbf{H}_{\mathrm{est}, i} ) \Big) \\ & + \sqrt{2\beta \cdot \alpha_i} \cdot \sigma_{z_i} \cdot \mathbf{\zeta},
\end{split}
\end{align}
\noindent where the term involving $\log p_Y ( \mathbf{Y} )$ does not depend on the current $\mathbf{H}_{\mathrm{est}, i}$ and has zero gradient.
\begin{figure*}[!t]
\centering
\includegraphics[width=0.8\linewidth]{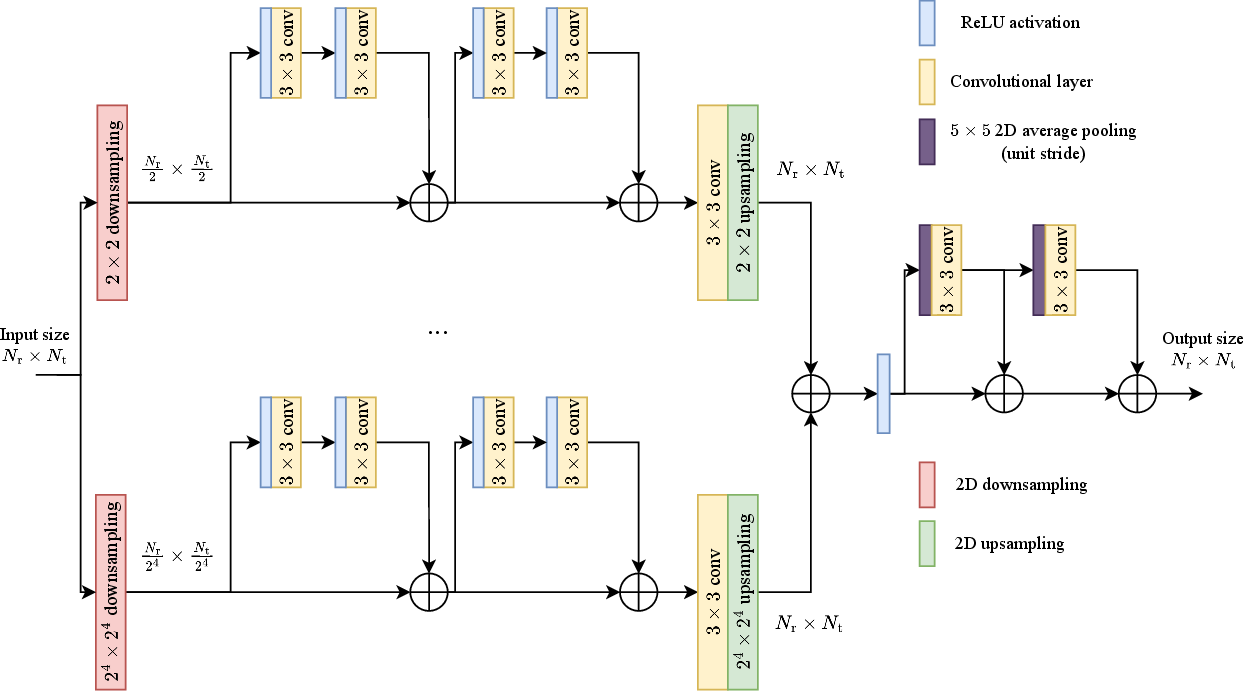}
\caption{Detailed block diagram of $s_\theta$ using the RefineNet architecture. The figure shows a core RefineNet block, that is serially repeated for $D$ times.}
\label{fig:score_architecture}
\end{figure*}
Using \eqref{eq:matrix_model} and the assumption that the pilot noise $\mathbf{N}$ is Gaussian, we obtain that the conditional distribution of $Y | H$ is Gaussian with mean $\mathbf{H}_{\mathrm{est},i} \mathbf{P}$ and covariance matrix $\sigma_\mathrm{pilot}^2 \mathbf{I}$, and its score can be derived in closed-form as:
\begin{equation}
    \nabla \log p_{Y | H} ( \mathbf{Y} | \mathbf{H}_{\mathrm{est}, i} ) = \frac{( \mathbf{H}_{\mathrm{est},i} \mathbf{P} - \mathbf{Y} ) \mathbf{P}^\mathrm{H}}{\sigma_\mathrm{pilot}^2}.
\label{eq:data_score}
\end{equation}
In practice, we also include an annealing term in the denominator of the above, as shown in Algorithm~\ref{alg:inference}. Finally, the only unknown component in \eqref{eq:abstract_update} is now the score of $p_H ( \mathbf{H} )$, evaluated at the current estimate. While this does not have a closed-form expression, we can leverage the score-based model $s_\theta$ trained using the procedure described in Section~\ref{sec:training}. This leads to the final channel estimation procedure in Algorithm~\ref{alg:sampling} and illustrated in Figure~\ref{fig:score_inference}, where, additionally, at each noise level we perform $M = 3$ updates \cite{song2020improved}. \revision{The update in \eqref{eq:abstract_update} can also be interpreted as a form of noisy, regularized gradient descent on the pilot consistency loss.}
\begin{algorithm}[!t]
\caption{MIMO Channel Estimation with Score-Based Generative Models.}\label{alg:sampling}
\begin{algorithmic}
\STATE \textbf{Inputs}: Pilot matrix $\mathbf{P}$, received pilots $\mathbf{Y}$, pretrained score-based model $s_\theta$, received noise power $\sigma_\mathrm{pilot}^2$, inference noise levels $\sigma_{z_i}^2$ (same as what $s_\theta$ was trained with), hyper-parameters $L, M, \alpha_0, \beta$ and $r < 1$.
\STATE \textbf{Generate random initial estimate}: $\mathbf{H}_\mathrm{est} \sim \mathcal{CN} (0, \mathbf{I} )$
\STATE \textbf{for} $i = 1, \ \dots, \ L$
\STATE \hspace{0.5cm} Set annealed noise level $\sigma \gets \sigma_{z_i}$.
\STATE \hspace{0.5cm} \textbf{for} $m = 1, \ \dots, \  M $
\STATE \hspace{1cm} Generate annealing noise $\mathbf{\zeta} \sim \mathcal{CN} (0, \mathbf{I} )$.
\STATE \hspace{1cm} $\mathbf{H}_\mathrm{est} \gets \mathbf{H}_\mathrm{est} + \alpha_0 \cdot r^i \cdot \frac{( \mathbf{H}_{\mathrm{est}} \mathbf{P} - \mathbf{Y}) \mathbf{P}^\mathrm{H}}{\sigma_\mathrm{pilot}^2 + \sigma^2} + \alpha_0 \cdot r^i \cdot s_\theta ( \mathbf{H}_\mathrm{est} ) + \sqrt{2 \beta \cdot \alpha_0 \cdot r^i} \cdot \sigma \cdot \mathbf{\zeta}$.
\STATE \textbf{Output}: Estimated channel matrix $\mathbf{H}_\mathrm{est}$.
\end{algorithmic}
\label{alg:inference}
\end{algorithm}

\section{Theoretical Results}
\label{sec:theory}
An important aspect of using posterior sampling is that previous work has derived theoretical guarantees for recovering the correct estimate (up to the ambient noise level), assuming that a sufficient number of measurements are available \cite{jalal2021instance,jalal2021robust}, even when posterior sampling is performed with respect to a mismatched distribution in terms of the 2-Wasserstein distance. In this section we derive an expression for the effects of train-test distributional mismatch for a simplified class of tapped delay line SISO channels.

Let $p$ and $r$ be two probability distributions, and $\mathcal{W}_2^2 (p, r) = \inf_{X \sim p, Y \sim r} \mathbb{E} \norm{X-Y}_2^2$ be the squared 2-Wasserstein distance between the two probability distributions \cite{panaretos2019statistical}. The infimum is taken over pairs of random variables with arbitrary joint distribution and marginals given by $p$ and $r$. If $\mathcal{W}_2^2 (p, r) > 0$, then the two distributions are mismatched.

We define the following family of sparse channels with complex-valued gains $\{ g_i \}_{i = 1, \dots, K}$ and real-valued delays $\{ \tau_i \}_{i = 1, \dots, K}$, and the channel gain at an arbitrary delay $t$ given by:
\begin{equation}
    h (t) = \sum_{i=1}^K g_i \cdot \delta(t - \tau_i).
\end{equation}
To obtain a vector channel, we sample $h(t)$ along the delay dimension, using a pre-determined, fixed waveform $w(t)$ at $N$ equally spaced points as $\mathbf{h}[i] = ( w(t) * g(t) ) ( iT )$, where $1/T$ is the sampling resolution. In practice, this corresponds to the sampled impulse response of a SISO channel. We make the following additional assumptions:
\begin{itemize}
    \item $\tau_i = -\alpha_i \log X_i$, where $\alpha_i$ is a constant and $X_i \sim \mathcal{U}(0, 1)$, $\forall i$. This is a realistic model for tap delays, and is used in the standardized CDL models \cite{3gpp.38.901}.
    \item $g_i \sim \mathcal{CN}(0, \sigma_i^2)$, $\forall i$, where $\sigma_i$ is a constant.
    \item $\sum_{i=1}^{K} \sigma_i^2 = 1$, i.e., the taps are normalized to average unit total power.
\end{itemize}
The above assumptions imply that a distribution of normalized channels $p_H$ is parameterized by the $2K$ real-valued degrees of freedom $\{ \sigma_i \}_{i = 1 \dots K}$ and $\{ \alpha_i \}_{i = 1 \dots K}$. For example, a line-of-sight (LOS) family of channels could exhibit $\alpha_1 < \alpha_i, \forall i > 1$, as well as $\sigma_1 \gg \sigma_i, \forall i > 1$, meaning that the channel profile contains a large magnitude path with a small delay, followed by delayed rays with much lower gains. Note that these assumptions only concern the marginal distributions of the taps and delays, while placing no restrictions on their joint distribution. Let $h_1$ and $h_2$ be two distributions following the previous assumptions, with degrees of freedom $\{\sigma_i^{(1)}, \alpha_i^{(1)}\}_{i = 1 \dots K}$ and $\{\sigma_i^{(2)}, \alpha_i^{(2)}\}_{i = 1 \dots K}$, respectively. The main theorem is given below.
\begin{theorem}
Let $\mathbf{h}^\star$ be a vector channel sampled from $h_2$. Then, posterior sampling with respect to $p_{h_1}$, using a number of $\mathcal{O} ( 1/\delta_\mathrm{MNR} )$ linear, Gaussian pilot measurements, at a noise level $\sigma_\mathrm{pilot}$, recovers $\mathbf{h}_\mathrm{est}$ such that:
\begin{equation}
    \norm{ \mathbf{h}^\star - \mathbf{h}_\mathrm{est}}_2 \le C \sigma_\mathrm{pilot} \quad \mathrm{w.p.} \quad 1 - \mathcal{O} (\delta_\mathrm{MNR}),
\end{equation}
\noindent where:
\begin{align}
\begin{split}
    \delta_\mathrm{MNR}^2  & = \frac{\mathcal{W}_2^2 ( h_1, h_2 )}{\sigma_\mathrm{pilot}^2} \\
    & \le \frac{\sum_i^K ( \sigma_i^{(1)} - \sigma_i^{(2)} )^2 + 2 ( \alpha_i^{(1)} - \alpha_i^{(2)} )^2}{\sigma_\mathrm{pilot}^2},
\label{eq:mnr}
\end{split}
\end{align}
\noindent and $C$ is a constant.
\label{theorem:two}
\end{theorem}
\begin{proof}
For the complete proof, see the Appendix. At a high-level, it involves two stages, in which the main novelty lies in the second stage: applying Theorem~1.1 from \cite{jalal2021instance} for the 2-Wasserstein distance, followed by deriving an upper bound for $\delta_\mathrm{MNR}^2$ given the previous assumptions.
\end{proof}

The quantity $\delta_\mathrm{MNR}^2$ represents the \textit{mismatch-to-noise} ratio, i.e., the ratio between the distributional mismatch of $h_1$ and $h_2$, and the pilot noise power. In general, $\delta_\mathrm{MNR}^2$ is a function of the two distributions. Theorem~\ref{theorem:two} states that the probability of successful channel estimation with posterior sampling is inversely proportional to $\delta_{\mathrm{MNR}}$, with a smaller $\delta_\mathrm{MNR}$ leading to higher probability of correct channel estimation. In general, this happens in two conditions: a very small mismatch $\mathcal{W}_2^2 ( h_1, h_2 )$ exists between the train and test distributions, or the noise power $\sigma_\mathrm{pilot}^2$ is very large. The above leads to the following performance analysis for channel estimation with posterior sampling:
\begin{enumerate}[(i)]
    \item When there is a train-test match, we have $\mathcal{W}_2^2 ( h_1, h_2 ) = 0$ and channel estimation with posterior sampling is optimal up to the noise level if using sufficient measurements.
    \item The probability of successful channel estimation decreases as the distributional distance between the training and test environments increases, at a fixed noise level $\sigma_\mathrm{pilot}^2$.
    \item Assuming a fixed $\mathcal{W}_2^2 ( h_1, h_2 )$, the probability of successful channel estimation \textit{increases} as the ambient noise level $\sigma_\mathrm{pilot}^2$ increases. While this implies favourable scaling in very low SNR, the score-based estimator in this regime is limited by the large noise power.
\end{enumerate}

\noindent \textbf{Applicability to MIMO Channels.} While the theory is derived for vector channels, the implications of Theorem~\ref{theorem:two} are verified in Section~\ref{sec:experiments} for MIMO channels, and are successful in predicting the performance of the proposed approach in varying test-time environments. In particular, we verify that channel estimation with posterior sampling, under moderate distributional shifts and with realistic, non-Gaussian transmitted pilot matrices verifies the three behaviours outlined in the previous paragraph.

As an alternative to delay-domain SISO channels, it is possible to interpret a vector channel as either a single-input multiple-output (SIMO) or multiple-input single-output (MISO) channel sampled along the spatial angular direction, using an array with a finite number of elements. As long as the distributions of the angles of arrival and departure allow for tractable computation of the 2-Wasserstein distance, a similar result could be derived for MIMO channels.

\section{Simulation Results and Discussion}
\label{sec:experiments}
We evaluate Algorithm~\ref{alg:sampling} in simulated settings, across a wide range of SNR values, using channel estimation fidelity, end-to-end error rates in a simulated coded system and \revision{training and inference complexity} to draw conclusions. To emulate deployment in a novel wireless propagation environment, Sections~\ref{sec:robust_est} and \ref{sec:end_to_end_exp} consider scenarios where models are tested on different distributions than the training one, without any prior knowledge about the test distribution or adaptation, \revision{and without any additional retraining}.

\subsection{Data and Training}
\begin{table*}
\begin{center}
\caption{Parameters used to simulate CDL channel realizations}
\label{table:data}
\begin{tabularx}{0.95\textwidth}{|c|*{5}{Y|}}
\hline
Channel Model & Carrier Frequency & Antenna Arrays & Antenna Spacing & ($N_\mathrm{r}, N_\mathrm{t}$) \\
\hline
Training: CDL-{B, C} & \multirow{2}{*}{$40$ GHz} & \multirow{2}{*}{ULA, UPA} & $\lambda/2$ - ULA & (16, 64), (32, 128) - ULA \\
Testing: CDL-A, B, C, D & & & ($\lambda/4, \lambda/4$) - UPA & (64, 256) - UPA \\
\hline
\end{tabularx}
\end{center}
\end{table*}
We use the CDL family of channel models to generate training, validation and test data. CDL-D channels are LOS, and generally the easiest to estimate due to their very sparse structure in the beamspace representation. CDL-B and -C channels are NLOS, while CDL-A channels have both components. For training, we use $10000$ channel realizations from a specific CDL model. To generate training samples, we increment the seed of the CDL generator, and pick the first subcarrier of the first symbol in each generated channel. Details about the used parameters are given in Table~\ref{table:data}. For tuning the hyper-parameter of CS methods we use a set of $100$ channel realization from the training distribution. Exact details on the training parameters (optimizer, learning rate, batch size) are available in the source code repository.

For testing, we generate a new set of $100$ channel realizations from each target distribution, using different random seeding than training and validation. For pilots $\mathbf{P}$, we use matrices of size $N_\mathrm{t} \times N_\mathrm{p}$ with randomly chosen QPSK elements (unit-power, two-bit phase-quantized random beamforming). We normalize all channels using the average channel power from the training set taken across all training samples and entries, and define the average SNR as $N_t / \sigma_\mathrm{pilot}^2$.

\subsection{Baselines}
The following baselines are evaluated and are available in the source code repository:
\begin{itemize}
    \item \revision{\textbf{Maximum Likelihood (ML)} \cite{nayebi2017semi}: This approach does not assume any prior information about the structure of $\mathbf{H}$ and aims to maximize the log-likelihood $p ( \mathbf{Y} | \mathbf{H} )$. Using knowledge that the noise $\mathbf{N}$ in \eqref{eq:basic_model} is Gaussian with zero mean and known variance $\sigma_\mathrm{pilot}^2$, and that the pilot entries are chosen i.i.d., yields the closed form solution via the regularized pseudo-inverse as \cite{nayebi2017semi}: $\mathbf{H}_\mathrm{ML} = \mathbf{Y} \mathbf{P}^\mathrm{H} ( \mathbf{P} \mathbf{P}^\mathrm{H} + \sigma_\mathrm{pilot}^2 \mathbf{I} )^{-1}$.}

    \item \textbf{Lasso} \cite{venugopal2017channel}: This is a CS-based approach that uses $\ell_1$-norm element-wise regularization in the two-dimensional Fourier (beamspace) domain. Channel estimation is formulated as the solution to the following optimization problem:
    \begin{equation}
        \argmin_\mathbf{H} \frac{1}{2} \norm{\mathbf{Y} - \mathbf{H} \mathbf{P}}_F^2 + \lambda \norm{\mathbf{F}_\mathrm{left} \mathbf{H} \mathbf{F}_\mathrm{right}^\mathrm{H}}_{1,1},
    \label{eq:cs_basic}
    \end{equation}
    \noindent where $\mathbf{F}_\mathrm{left}$ and $\mathbf{F}_\mathrm{right}$ are square, discrete Fourier matrices of size $N_\mathrm{r} \times N_\mathrm{r}$ and $N_\mathrm{t} \times N_\mathrm{t}$, respectively. We use gradient descent with momentum to solve \eqref{eq:cs_basic} and we tune the step size, number of optimization steps, and the value of $\lambda$ using the validation set.
    
    \item \revision{\textbf{EM-GM-AMP} \cite{schniter2014channel}: This approach uses a Gaussian mixture (GM) prior in the beamspace domain of $\mathbf{H}$, justified by the sparse angular nature of mmWave propagation channels \cite{heath2016overview}. We use the publicly available Matlab implementation provided in \cite{emgmamp_code} and tune the number of internal expectation maximization (EM) steps, as well as the stopping condition using the validation set.}

    \item \textbf{fsAD} \cite{bhaskar2013atomic,zhang2017atomic}: This represents a \revision{classical} CS-based approach for recovering channel matrices assumed sparse in the continuous spatial frequency domain, similar to Newtonized OMP (Orthogonal Matching Pursuit) \cite{mamandipoor2016newtonized}. In scenarios with uniform linear arrays (ULA) at both the receiver and transmitter, we use a formulation that \revision{assumes sparsity in the over-sampled beamspace domain} \cite{zhang2017atomic}:
    \begin{equation}
    \label{eq:fast_atomic}
        \argmin_\mathbf{H} \frac{1}{2} \norm{\mathbf{Y} - \mathbf{H} \mathbf{P}}_F^2 + \lambda \norm{\mathbf{W}_\mathrm{left} \mathbf{H} \mathbf{W}_\mathrm{right}^\mathrm{H}}_{1,1},
    \end{equation}
    \noindent where $\mathbf{W}_\mathrm{left}$ and $\mathbf{W}_\mathrm{right}$ represent $l$ times over-sampled DFT matrices. We use $l=4$, and tune the parameter $\lambda$, step size and number of steps using the validation set.

    \item \textbf{WGAN} \cite{balevi2020high}: This represents an unsupervised framework for using generative adversarial networks in channel estimation. A generative model $g_\theta$ is first trained to map low-dimensional vectors $\mathbf{z}$ to channel matrices $\mathbf{H}$, using an adversarial loss and regularization \cite{gulrajani2017improved}. During inference, a similar optimization problem to ours is formulated and solved iteratively, by inverting the generative model. We tune the hyper-parameters using the validation set.

    \item \textbf{L-DAMP} \cite{he2018deep,metzler2017learned}: This represents a powerful data-driven algorithm that uses deep unrolling and end-to-end learning. We use a denoising convolutional neural network (DnCNN) \cite{zhang2017beyond} backbone, and train a separate L-DAMP model for each value of $\alpha$.
    
    \item \revision{\textbf{Approximate MMSE}: By definition, the minimum mean squared error (MMSE) estimator is given by $\mathbb{E}[ \mathbf{H} | \mathbf{Y} ]$. Because the log-prior for realistic channel models, such as CDL channels, is generally intractable and does not have a closed form expression (e.g., the tap locations themselves are stochastic, which would require additional assumptions and factorization of the prior to become tractable \cite{neumann2018learning}), we leverage score-based models to obtain an empirical upper performance bound by averaging a number of samples from the estimated posterior (given by running Algorithm~\ref{alg:sampling}) $p ( \mathbf{H} | \mathbf{Y} )$ \cite{jalal2021robust}. For the remaining simulations, we use $50$ posterior samples, each obtained from a different run of Langevin dynamics with different initial estimates.}
\end{itemize}
\begin{figure*}[!t]
\centering
\includegraphics[width=0.9\linewidth]{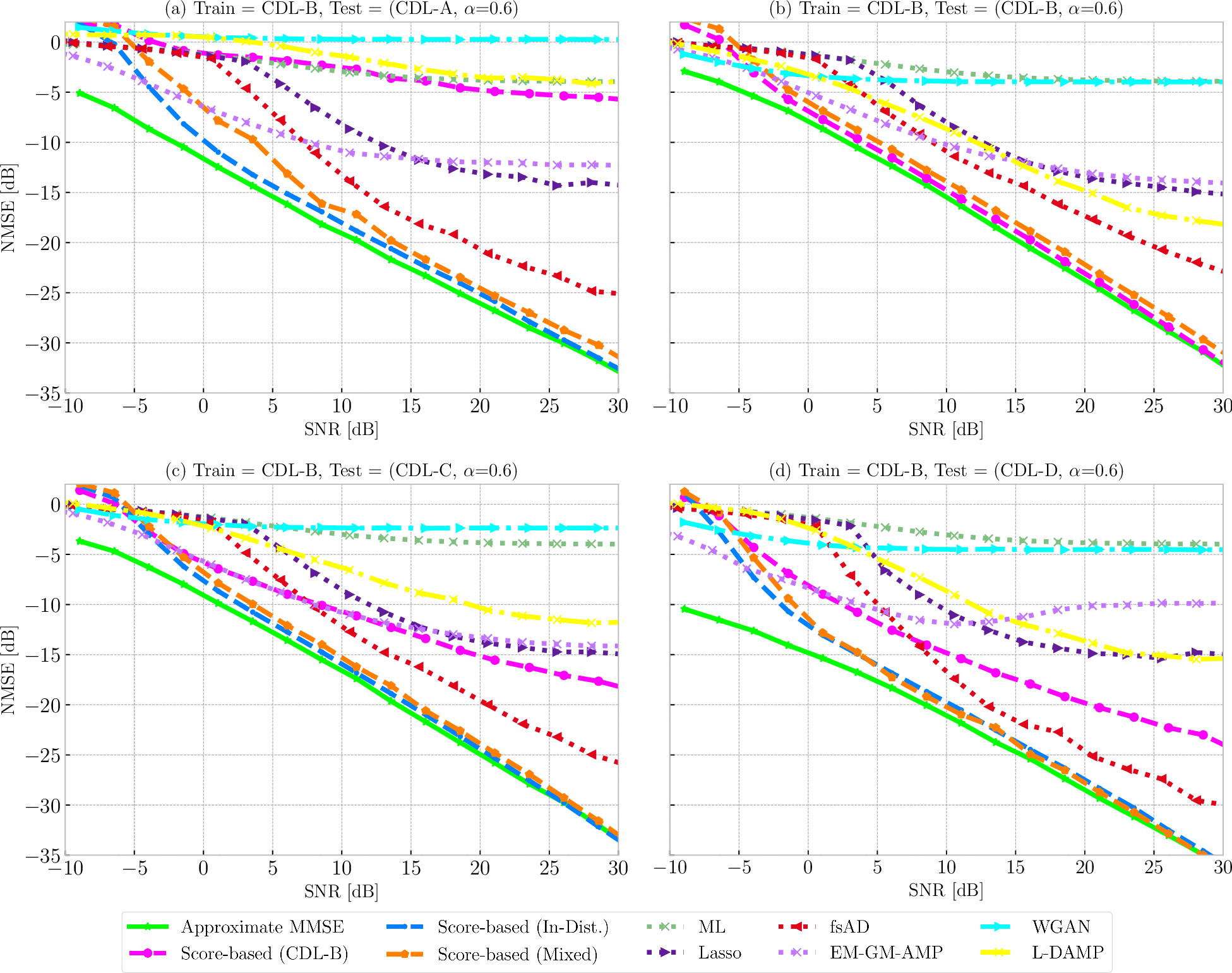}
\caption{Channel estimation performance using methods trained on CDL-B channels in a $16 \times 64$ mmWave MIMO scenario. Each of the four sub-figures shows performance in a different CDL environment, when $\alpha = 0.6$ ($38$ pilot vectors).}
\label{fig:cdl_b_results}
\end{figure*}
\subsection{Robust Estimation Performance}
\label{sec:robust_est}
We train score models on $N=10000$ samples of CDL-B and CDL-C channels separately, and test them in CDL-\{A, B, C, D\} environments, without fine-tuning or adaptation. As a reference, we also train a \textit{Mixed} score-based model, which uses $10000$ training samples each from all four CDL-\{A, B, C, D\} models. \revision{We also train score-based models on their matching test-time distribution to evaluate the potential gains of adapting models. In practice, this adaptation would be done at the base station or user equipment, after a training set from an environment is collected.} We measure estimation quality using the normalized mean squared error (NMSE), defined as:
\begin{equation}
    \mathrm{NMSE \ [dB]} = 10 \log_{10} \frac{\norm{\mathbf{H}_\mathrm{est} - \mathbf{H}}_F^2}{\norm{\mathbf{H}}_F^2}.
\end{equation}
Figure~\ref{fig:cdl_b_results} shows estimation results when using models trained on CDL-B channels. In the in-distribution setting score-based models recover channels up to the noise level, and surpass prior work by at least $3$ dB in NMSE for SNR values between $-5$ and $30$ dB. In all test conditions, the \textit{Mixed}, \revision{in-distribution and approximate MMSE methods} scale favourably with high SNR. \revision{ML, EM-GM-AMP and Lasso are competitive in the very low SNR regime, but saturate performance quickly. This makes score-based models an attractive candidate for scenarios where very accurate channel estimation is desired, such as fixed, broadband wireless access.} The slight drop in performance between the \textit{Mixed} model and the model trained only on CDL-B channels in Figure~\ref{fig:cdl_b_results}(b) is owed to the finite capacity of the deep network, performing slightly poorer on a specific distribution when data from other distributions is included in training.

The WGAN approach can competitively estimate channels in the very low SNR regime, but, in general, suffers from saturation of the performance at SNR values greater than $-5$ dB. \revision{This has been previously observed in \cite{bora2017compressed} and is owed to the sub-optimal approach (in practice, using an Adam \cite{kingma2014adam} optimizer) used to invert the latent representation, whereas a score-based model using Langevin dynamics converges to a near-optimal solution.} Estimating CDL-A channels at a pilot density $\alpha=0.6$ using a score-based model trained on CDL-B is difficult, as shown in Figure~\ref{fig:cdl_b_results}(a). \revision{Because of CDL-A channels including both LOS and non-LOS components, this leads to a large distributional distance between CDL-A and all other models, where a subset of propagation paths will always be missing. In accordance with Theorem~\ref{theorem:two}, this causes an error floor in the high SNR regime.}
\begin{figure*}[!t]
\centering
\includegraphics[width=0.9\linewidth]{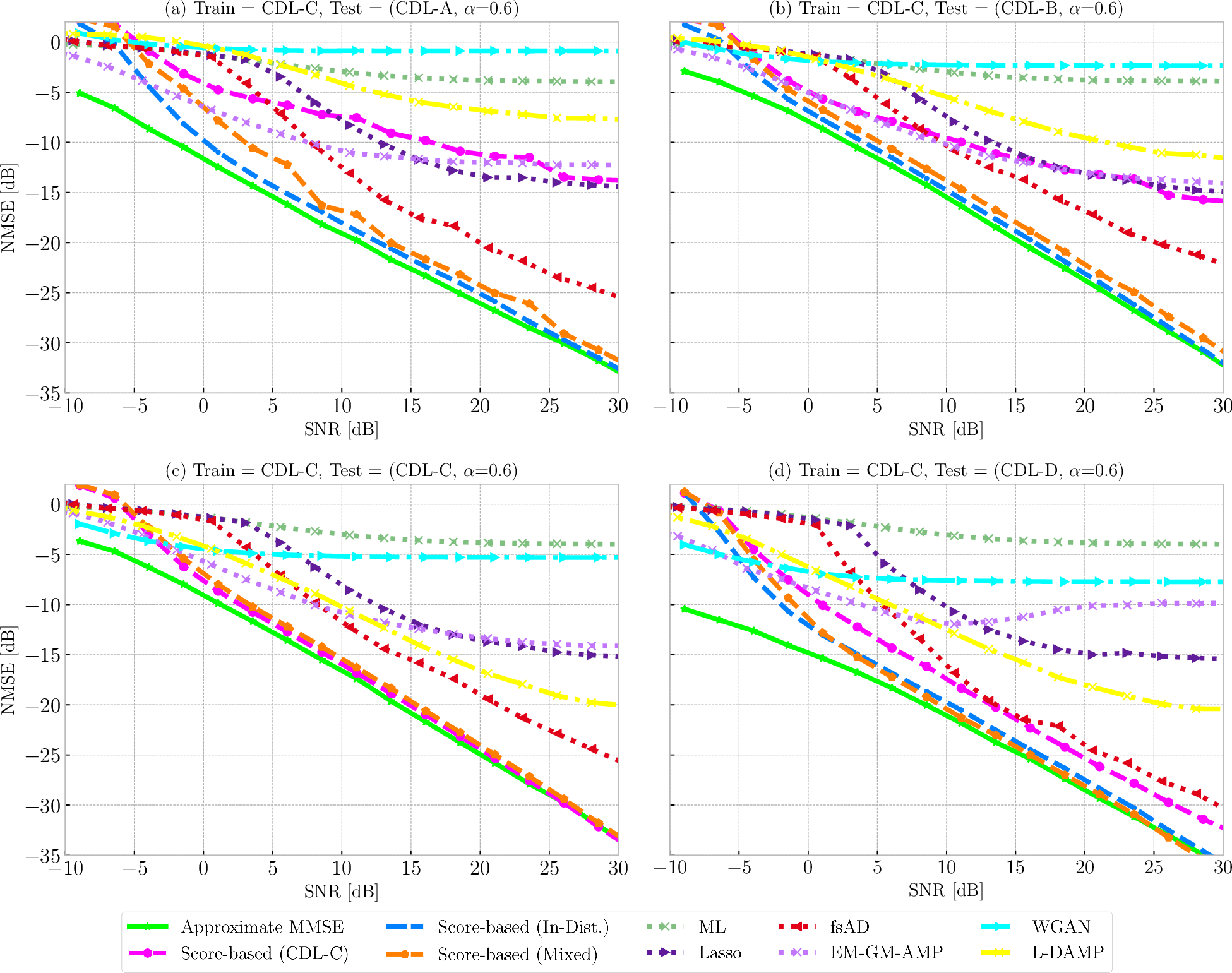}
\caption{Channel estimation performance using methods trained/tuned on CDL-C channels in a $16 \times 64$ mmWave MIMO scenario. Each of the four sub-figures shows performance in a different CDL environment, when $\alpha = 0.6$ ($38$ pilot vectors).}
\label{fig:cdl_c_results}
\end{figure*}
Figure~\ref{fig:cdl_c_results} plots estimation performance with models trained only on CDL-C channels, and tested on all four CDL models. This generally corroborates the findings from Figure~\ref{fig:cdl_b_results}: score-based models recover channels up to the noise floor in-distribution (Figure~\ref{fig:cdl_c_results}(c)), and outperform all baselines for SNR values between $-5$ and $30$ dB. A notable difference here is the much better generalization capability of score-based models from CDL-C to CDL-D environments, as observed in Figure~\ref{fig:cdl_c_results}(d). This is explained by the fact that CDL-C channels contain the least amount of scattering, and are most similar to CDL-D channels, leading to a lower $\delta_\mathrm{MNR}(\mathrm{CDL-C}, \mathrm{CDL-D})$, in line with the Theorem~\ref{theorem:two}. Another difference that occurs when training on CDL-C channels, is that generalization to CDL-A models is also improved, matching the performance of the Lasso algorithm, as shown in Figure~\ref{fig:cdl_c_results}(a).

In both settings of Figure~\ref{fig:cdl_b_results} and Figure~\ref{fig:cdl_c_results}, we also find that L-DAMP has similar scaling in terms of SNR, and is a reliable estimation approach, especially for the simpler CDL-C and CDL-D models, as highlighted in Figures~\ref{fig:cdl_b_results}(c) and \ref{fig:cdl_b_results}(d). The Lasso \revision{and EM-GM-AMP} approaches are also competitive, in general, and surpass score-based models (trained on CDL-B), WGAN, and L-DAMP when evaluated on CDL-A channels, as shown in Figure~\ref{fig:cdl_b_results}(a).

The findings in this section support the conclusions of Theorem~\ref{theorem:two}:
\begin{enumerate}[(i)]
    \item Posterior sampling with score-based models can recover channels up to the noise floor when the test distribution matches the training and the SNR is above $0$ dB.
    \item Channel estimation with posterior sampling outperforms the baselines in the low SNR regime ($-5$ to $0$ dB), under test-time distributional shifts.
    \item Test channel distributions that are significantly more different than the training distributions lead to error floors in the high SNR regime (Figures~\ref{fig:cdl_b_results}(a) and \ref{fig:cdl_c_results}(a)).
\end{enumerate}
Finally, it is worth highlighting that \revision{CS approaches} are overall robust and retain performance under distributional shifts, regardless of what environment the hyper-parameters are tuned on. While these algorithms are useful in idealized settings, the following limitations remain:
\begin{itemize}
    \item When deriving the approximation in \eqref{eq:fast_atomic}, we use knowledge about the shapes of the antenna arrays at the receiver and transmitter. In general settings, this information may not be available, or the array response may be intractable (e.g., for non-uniform arrays). In contrast, score-based generative models learn the distribution of the channels while \textit{implicitly} learning the array configuration, with no external knowledge required.
    \item A key assumption made by \revision{CS approaches} is that channels are exactly sparse (on a discrete grid \revision{for Lasso and EM-GM-AMP}, or on the continuum of spatial frequencies \revision{for fsAD}). While this holds for simulated channel using the CDL models, it is generally not true in practice, e.g., as shown in channel sounding experiments at mmWave \cite{rappaport2012broadband,molisch2016millimeter}. Score-based models are completely data-driven, with no assumptions required.
\end{itemize}

\subsection{End-to-End Coded Performance}
\label{sec:end_to_end_exp}
To consider a realistic performance metric, in this section we evaluate end-to-end coded bit error rates in a simulated downlink physical layer flow of cellular communication systems. A block diagram of the system is shown in Figure~\ref{fig:e2e_diagram}. The transmitter (base station) sends a pilot matrix $\mathbf{P}$ across multiple channel uses, which is received by the user as $\mathbf{Y}$, using the model in \eqref{eq:matrix_model}. Channel estimation is used to obtain $\mathbf{H}_\mathrm{est}$.
\begin{figure*}[!t]
\centering
\includegraphics[width=0.9\linewidth]{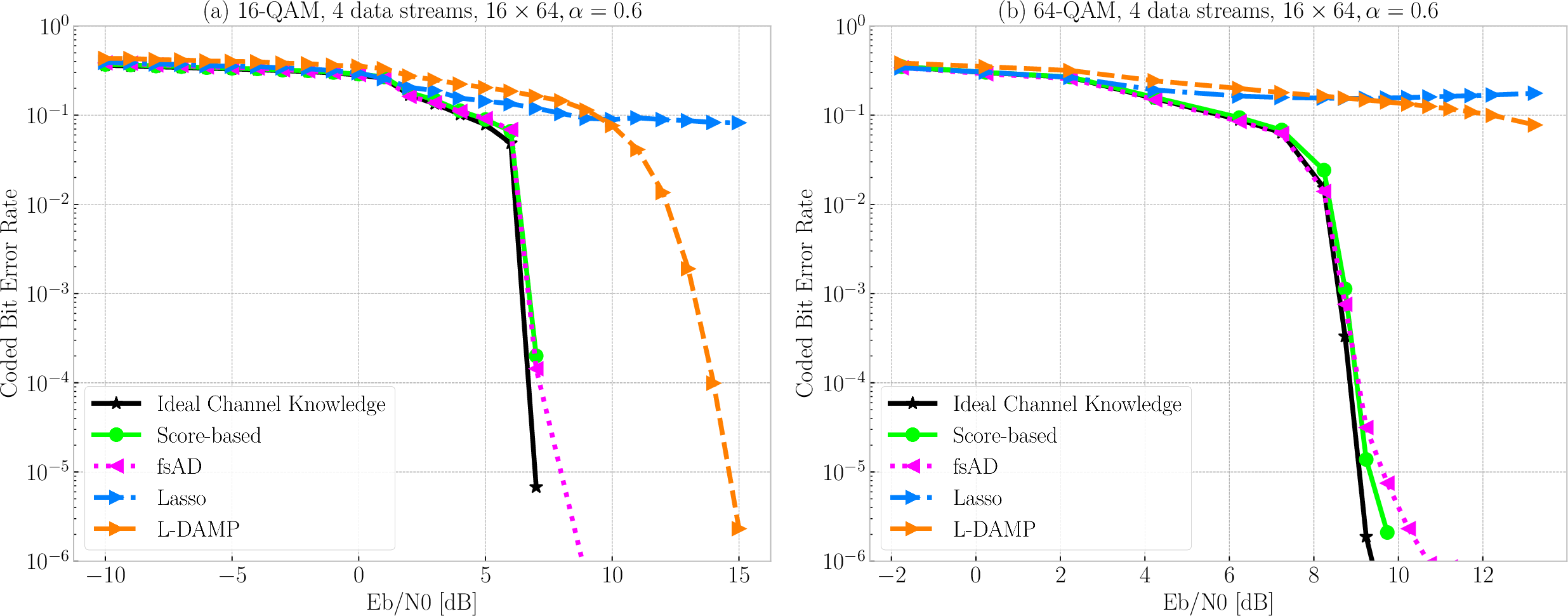}
\caption{End-to-end coded bit error rate as a function of SNR in CDL-D out-of-distribution channels, with $\alpha = 0.6$. Methods are trained (Score-based, L-DAMP) or tuned (fsAD, Lasso) exclusively on CDL-C channels, and tested on CDL-D using: (a) 16-QAM and (b) 64-QAM data modulation.}
\label{fig:e2e_results}
\end{figure*}
\begin{figure}[!t]
\centering
\includegraphics[width=1\linewidth]{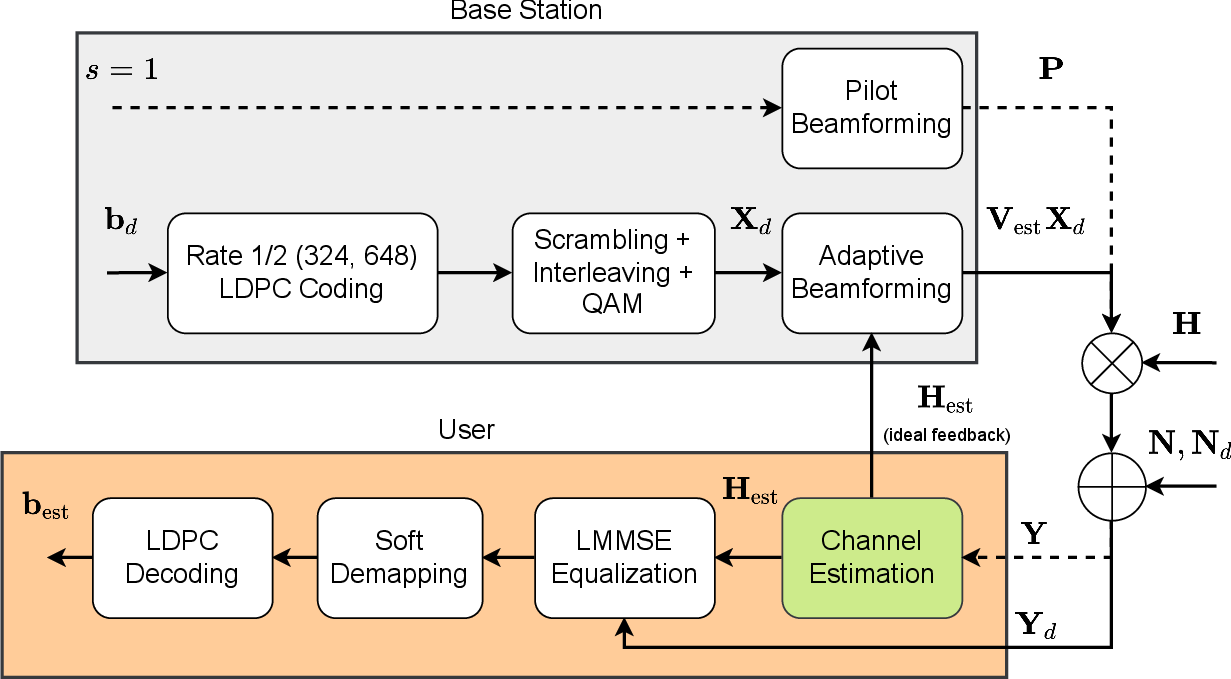}
\caption{Block diagram of the simulated setup used to evaluate end-to-end performance.}
\label{fig:e2e_diagram}
\end{figure}
For the data transmission, a payload of $324$ bits is encoded using a rate $1/2$ low-density parity-check (LDPC) code, followed by digital modulation. Symbols are split in groups of $N_\mathrm{s} = 4$ data streams, followed by adaptive digital beamforming with the unit-power $\mathbf{V}_\mathrm{est}$ matrix to obtain a matrix of $64$-dimensional column vectors. We assume that the channel is symmetric (e.g., time-division duplex), and that the transmitter obtains noiseless feedback about the estimated channel. The adaptive transmitter beamforming block uses the first four columns of $\mathbf{V}$ from the singular value decomposition of $\mathbf{H}_\mathrm{est}$ as beamforming weights. The data symbols are observed at the receiver as $\mathbf{Y}_d = \mathbf{H} \mathbf{V}_\mathrm{est}$ and are succeeded by linear MMSE detection to yield the equalized data symbols:
\begin{equation}
    \mathbf{X}_{d,\mathrm{eq}} = ( \mathbf{H}_\mathrm{est}^\mathrm{H} \mathbf{H}_\mathrm{est}^{} + \sigma_\mathrm{pilot}^2 \mathbf{I} )^{-1} \mathbf{H}_\mathrm{est}^\mathrm{H} \mathbf{Y}_d.
\end{equation}
The above is followed by entry-wise soft de-mapping, reshaping to a soft codeword, and LDPC decoding, to obtain the decoded bit stream $\mathbf{b}_\mathrm{est}$. This is repeated for two million codewords, each using a different CDL-D channel realization from a randomly generated test set, while using models trained on the CDL-C distribution for channel estimation, matching the setting of Figure~\ref{fig:cdl_c_results}(d).
Figure~\ref{fig:e2e_results} plots the performance results for four methods, as well as communication using ideally known channels, measured in terms of the energy per bit to noise ratio ($E_b/N_0$) of the corresponding channel estimation SNR. We omit WGAN, \revision{ML and EM-GM-AMP} from this investigation due to their estimation error floor, which cannot support packet communication at the considered coding rate. \revision{We also omit the approximate MMSE due to its large computation complexity, and instead use ideal channel knowledge as an upper performance bound.}

Figure~\ref{fig:e2e_results}(a) uses 16-QAM modulation. Even in this case, channel estimation errors lead to an error floor for the Lasso algorithm, while score-based models, fsAD, and L-DAMP manage to overcome this and decay the bit error rate at high $E_b/N_0$. In this case, an NMSE between $-20$ dB and $-15$ dB is required to avoid this error floor, that Lasso \revision{and other simple CS methods} cannot achieve, as per Figure~\ref{fig:cdl_c_results}(d). Figure~\ref{fig:e2e_results}(a) also shows that both score-based models and fsAD are competitive when compared to L-DAMP: the same coded bit error rate can be achieved with an $E_b/N_0$ that is $5$ dB lower than L-DAMP.
\begin{figure*}[!t]
\centering
\includegraphics[width=0.9\linewidth]{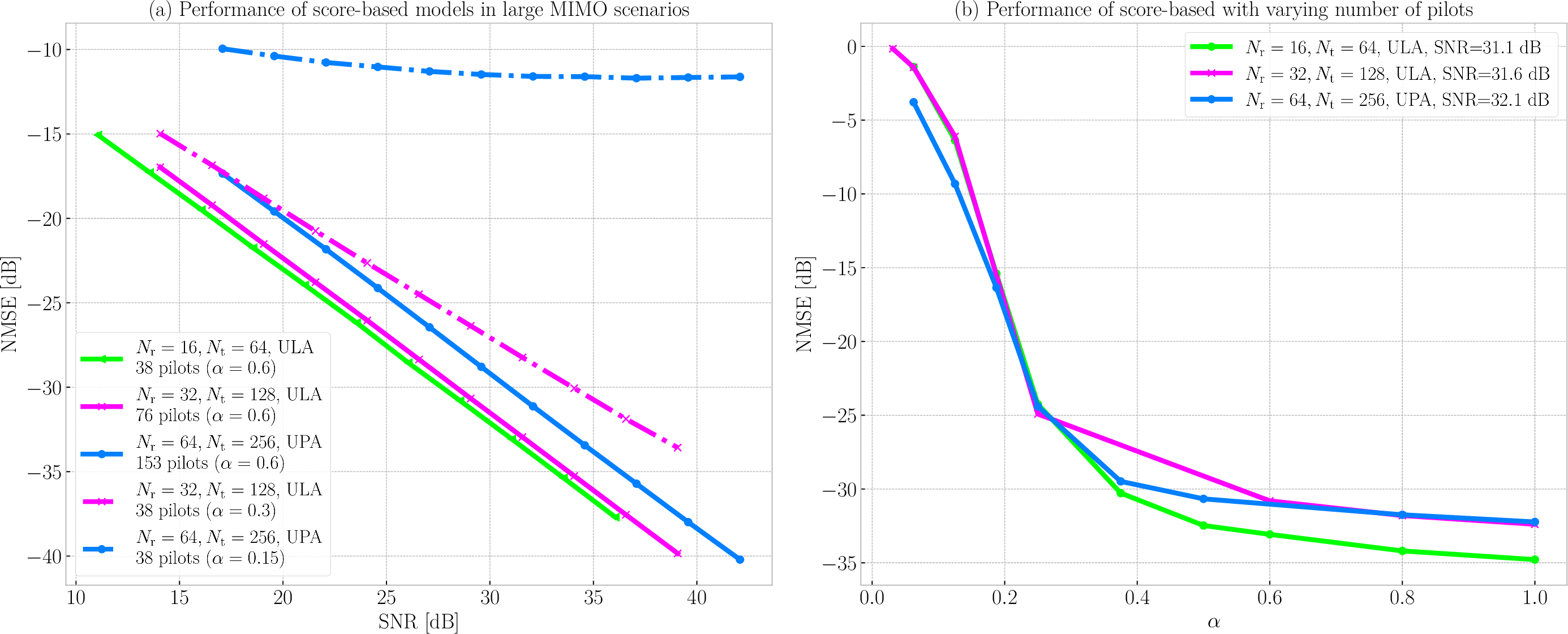}
\caption{Estimation performance of score-based models trained and tested on CDL-C channels at: (a) different MIMO sizes and (b) values of $\alpha$.}
\label{fig:size_scaling}
\end{figure*}
In Figure~\ref{fig:e2e_results}(b), the modulation is 64-QAM and accurate channel estimation is required for equalization and precoding. The L-DAMP approach is not enough to avoid an end-to-end error floor, while score-based models and fsAD overcome this floor and decay the bit error rate at approximately the same rate as ideal channel knowledge. Finally, the benefits of extremely accurate channel estimation at high SNR values (e.g., estimation NMSE lower than $-30$ dB at pilot SNR larger than $25$ dB) are also illustrated in Figure~\ref{fig:e2e_results}(b), where it can be seen that both methods depart from optimal performance, and score-based models improve performance by up to $0.5$ dB in $E_b/N_0$. \revision{Overall, these results highlight the importance of accurate channel estimation and how score-based models are competitive for this purpose and trade off increased computational complexity for accuracy.}

\subsection{Scaling to Large Channel Sizes}
We verify that channels can be estimated without error floors using a small score-based model, even for large channel sizes. Figure~\ref{fig:size_scaling}(a) shows the results of the experiment when training and testing score-based models for three channel sizes, where for each size we train a separate model. It can be seen that, given a fixed pilot overhead $\alpha$, there is a slight drop in performance at larger channel sizes -- we attribute this to using a score-based model of the same size (depth $D=6$, width $W=12$), regardless of channel size. To compare performance under a resource-limited scenario, we simulate larger sizes, where only $38$ pilot vectors are allowed, leading to low values of $\alpha$. In this case, estimation for larger sizes is competitive given sufficient pilots, but fails for $64\times 256$ with $\alpha = 0.15$.

Figure~\ref{fig:size_scaling}(b) investigates in more detail how performance scales with $\alpha$ for all three sizes in the high SNR regime, and the same model size as in Figure~\ref{fig:size_scaling}(a): here, we consistently find that there is a breaking point for $\alpha \approx 0.25$, below which the estimation error increases rapidly. This is indicative of not having sufficient measurements to meet the conditions of Theorem~\ref{theorem:two}.

\subsection{Complexity Analysis and Ablation}
\label{sec:ablation}
\begin{table*}
\begin{center}
\caption{Network size ablation results. $D$ represents the network depth \revision{(number of RefineNet blocks in Figure~\ref{fig:score_architecture})}, while $W$ is the number of hidden channels in the first hidden layer. Score-based models are trained on CDL-C channels and evaluated with $\alpha=0.6$ and low ($8$ dB) or high ($28$ dB) SNR.}
\label{table:ablation}
\begin{tabularx}{0.95\textwidth}{|c|*{10}{Y|}}
\hline
& $D=4$ & $D=5$ & $D=6$ & $D=4$ & $D=5$ & $D=6$ & $D=4$ & $D=5$ & $D=6$ \\
& $W=6$ & $W=6$ & $W=6$ & $W=12$ & $W=12$ & $W=12$ & $W=24$ & $W=24$ & $W=24$ \\
\hline 
\makecell{NMSE CDL-C [dB]} & $-37.8$ & $-38.6$ & $-38.7$ & $-37.3$ & $-38.5$ & $-38.8$ & $-37.8$ & $-38.5$ & $-38.7$ \\ 
\hline
\makecell{NMSE CDL-D [dB]} & $-36.3$ & $-34.8$ & $-34.8$ & $-36.3$ & $-35.3$ & $-35.2$ & $-36.0$ & $-35.2$ & $-35.5$ \\ 
\hline
Num. weights & $64$k & $184$k & $208$k & $263$k & $734$k & $828$k & $1046$k & $2938$k & $3314$k \\
\hline
Latency/step [ms] & $4.0$ & $5.2$ & $6.0$ & $4.1$ & $5.3$ & $6.0$ & $4.2$ & $5.4$ & $6.1$ \\
\hline
\makecell{Num. steps (low)} & $463$ & $580$ & $293$ & $287$ & $386$ & $326$ & $406$ & $388$ & $246$ \\
\hline
\makecell{Num. steps (high) } & $2034$ & $2121$ & $2119$ & $2016$ & $1935$ & $2071$ & $1908$ & $1931$ & $2045$ \\
\hline
\end{tabularx}
\end{center}
\end{table*}
All score-based models up to this point have used a RefineNet with $W = 24$ hidden channels in the first layer, and a depth $D$ of six residual blocks in both the encoder and decoder paths. In Table~\ref{table:ablation}, we investigate validation performance at SNR $28$ dB for nine model sizes that vary in depth and width. \revision{We chose this value as it corresponds to the high SNR regime, where score-based models have NMSE lower than $-30$ dB and performance is bounded by the quality of the learned prior, as indicated by Theorem~\ref{theorem:two}.} Our findings indicate that model performance is much more sensitive to depth rather than width, but even a shallow model (four blocks), has a performance loss of at most $1.5$ dB. Shallower models improve performance in the out-of-distribution setting, indicating that they overfit less. Table~\ref{table:ablation} also measures per-step latency on a machine with an NVIDIA A100 GPU and Intel Xeon 6230 CPU, as well as the optimal number of inference steps at two SNR values and $\alpha=0.6$. We find that depth is a main contributor to latency, whereas width is less impactful, due to the powerful parallelization capabilities of the GPU.

\revision{Figure~\ref{fig:interf_and_conv}(a) plots the convergence of the training loss in \eqref{eq:batch_loss} for score-based models trained on all channel distributions (note that the optimal value of this loss is not zero \cite{song2020improved}). It can be noticed that in all cases, convergence is stable and achieved in at most $60$ epochs using $N=10000$ training samples, with approximately one minute per epoch of training on the GPU. This is also a strong indicator that adapting a pre-trained score-based model would converge in at most the same time, and achieve near optimal performance in as few as ten epochs, given a large training set is available.}

\revision{Figure~\ref{fig:interf_and_conv}(b) plots the convergence of the Langevin dynamics updates in \eqref{eq:abstract_update} as a function of the number of steps, SNR, and the hyper-parameter $\beta$, averaged over the validation set. It can be seen that convergence is stable for all SNR levels, and that robust early stopping points across all validation samples can be found -- these benefit both complexity (reduced number of inference steps and up to $6\times$ reduced inference complexity) and performance. In general, lowering $\beta$ converges faster and to a better solution, but is also more sensitive to the stopping criterion. In all cases, the estimation loss during inference exhibits two stages: a fast convergence stage (in the first $200$ steps, where an estimate with an error less than $-10$ dB can be obtained for all SNR values) in which the residual errors are non-Gaussian and the updates lead to higher density regions of the distribution, and a slower convergence stage, where finer structural details of the estimated channel are captured.}
\begin{figure*}[!t]
\centering
\includegraphics[width=0.9\linewidth]{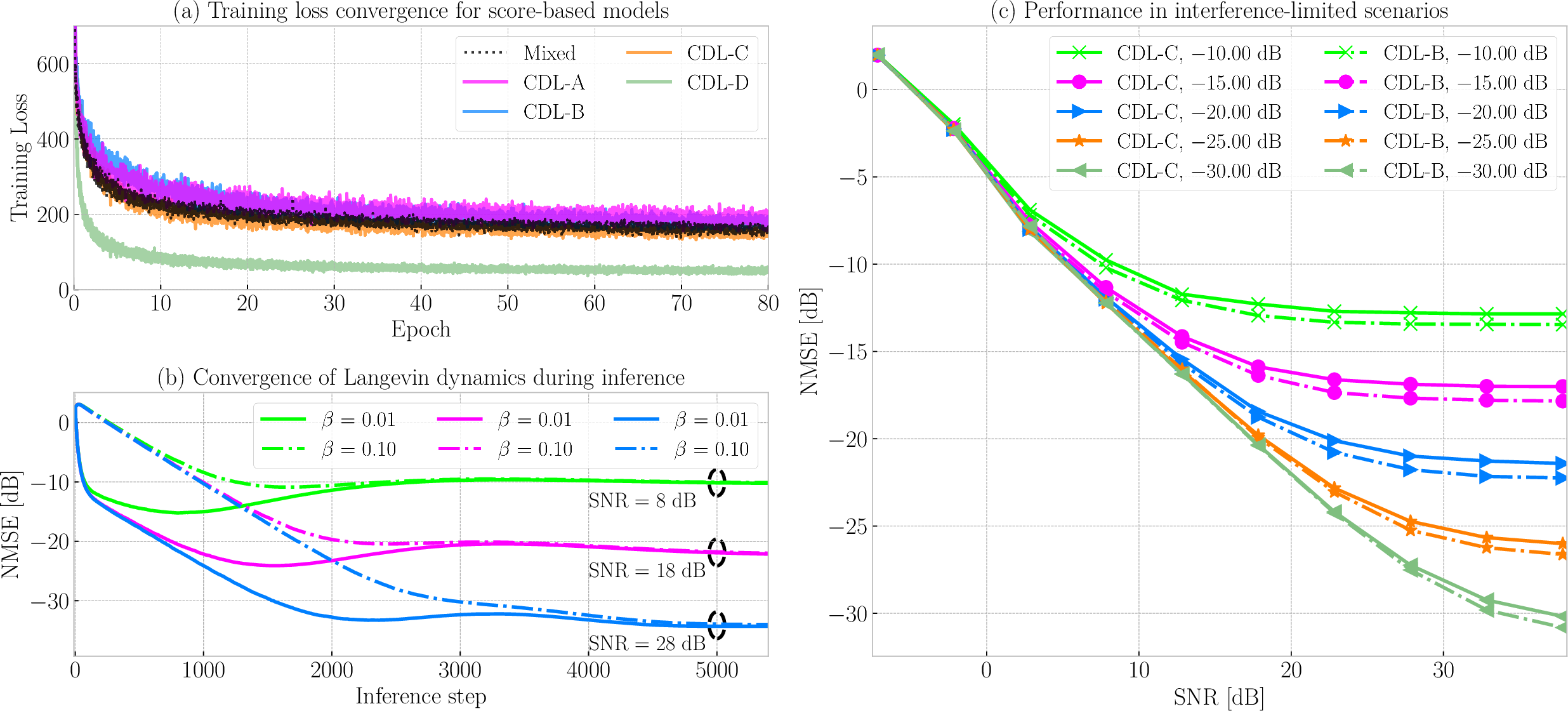}
\caption{\revision{(a) Convergence of the training loss for models with $D=6, W=24$. (b) Convergence of annealed Langevin dynamics for three SNR levels and two values of the $\beta$ hyper-parameter, averaged across $100$ validation samples. (c) In-distribution performance of score-based models in interference channels, with varying interference power.}}
\label{fig:interf_and_conv}
\end{figure*}
\begin{table*}
\begin{center}
\caption{Complexity, latency and memory footprint (active usage and model size) for different MIMO scenarios. Values are reported for estimation at SNR = $8$ dB, CDL-C channels, and $\alpha = 0.6$.}
\label{table:complexity}
\begin{tabularx}{0.95\textwidth}{|c|*{12}{Y|}}
\hline
& \multicolumn{3}{c|}{FLOP Count [GFLOPs]} & \multicolumn{3}{c|}{Latency [ms]} & \multicolumn{3}{c|}{Memory [MB]} & \multicolumn{3}{c|}{Model Size [kB]} \\ 
\hline
\backslashbox{Method}{$N_\mathrm{r}$} & $64$ & $128$ & $256$ & $64$ & $128$ & $256$ & $64$ & $128$ & $256$ & \multicolumn{3}{c|}{All} \\ 
\hline
Score-based & $18.9$ & $75.9$ & $303.5$ & $1500$ & $1770$ & $2070$ & $1.12$ & $4.46$ & $17.84$ & \multicolumn{3}{c|}{$828$} \\ 
\hline
ML & $0.001$ & $0.007$ & $0.061$ & $0.37$ & $0.62$ & $1.25$ & $0.06$ & $0.26$ & $1.04$ & \multicolumn{3}{c|}{n/a} \\ 
\hline
EM-GM-AMP & \multicolumn{3}{c|}{$\mathcal{O}(N_\mathrm{t} N_\mathrm{r})$} & $504$ & $4310$ & $>10$s & \multicolumn{3}{c|}{$\mathcal{O}(N_\mathrm{r}^2)$} & \multicolumn{3}{c|}{n/a} \\ 
\hline
fsAD & $6.86$ & $54.9$ & $439.4$ & $2590$ & $2600$ & $2700$ & $0.07$ & $0.29$ & $1.18$ & \multicolumn{3}{c|}{n/a} \\ 
\hline
Approx. MMSE & $56.9$ & $274.1$ & $1095$ & $2090$ & $3220$ & $9170$ & $19.4$ & $40$ & $122$ & \multicolumn{3}{c|}{$828$} \\ 
\hline
\end{tabularx}
\end{center}
\end{table*}
Table~\ref{table:complexity} evaluates the computational complexity of score-based generative models, \revision{approximate MMSE, and the compressed sensing baselines. It can be seen that inference latency is comparable to fsAD, and much lower than that of approximate MMSE and EM-GM-AMP in large MIMO scenarios. For EM-GM-AMP, we do not include exact FLOP counts and memory footprints due to lack of support in the MATLAB profiler. The FLOP count of score-based models scales much more favourably compared to fsAD due to being a purely convolutional algorithm, and has a performance loss smaller than $0.1$ dB in-distribution compared to the approximate MMSE algorithm, while requiring four times fewer FLOPs.}

\subsection{\revision{Performance in Interference Channels}}
\revision{We test estimation performance using pretrained score-based models on interference limited scenarios, which can become the primary bottleneck in mmWave cellular systems with dense deployments \cite{andrews2016modeling}. The system equation in \eqref{eq:basic_model} takes the form \cite{akoum2012coverage}:}
\begin{equation}
    \revision{\mathbf{Y} = \mathbf{H}\mathbf{P} + \gamma \mathbf{H}_\mathrm{I}\mathbf{P}_\mathrm{I} + \mathbf{N},}
    \label{eq:interference_model}
\end{equation}
\noindent \revision{where $\mathbf{H}_\mathrm{I}$ and $\mathbf{P}_\mathrm{I}$ are the interference channels and transmitted pilots, respectively, and $\gamma$ represents the power of the interference signal.}

\revision{We investigate a $64 \times 16$ MIMO scenario where $\mathbf{H}$ is sampled from CDL-C channels, and a score-based model trained on CDL-C models is used. $\mathbf{H}_\mathrm{I}$ and $\mathbf{P}_\mathrm{I}$ are sampled independently of $\mathbf{H}$ and $\mathbf{P}$, respectively, with $\mathbf{H}_\mathrm{I}$ sampled from either CDL-C or CDL-B channels. The value of $\gamma$ is varied from $-10$ dB to $-30$ dB. When running annealed Langevin dynamics, no knowledge about the distribution of interference or $\gamma$ is assumed, and the interference is treated as noise.} \revision{The results in Figure~\ref{fig:interf_and_conv}(b) show that estimation performance is reliable in noise-limited scenarios (without departing from the performance in Figure~\ref{fig:cdl_c_results}(c)) and that score-based models can effectively accommodate non-Gaussian noise scenarios. At higher SNR values, the performance becomes interference limited, with the performance floor improving relative to $\gamma$ by at least $2$ dB, experimentally showing that part of the interference signal can be reliably eliminated using score-based models.}

\section{Conclusion}
\label{sec:conclusion}
In this paper, we have introduced an unsupervised, probabilistic approach for MIMO channel estimation using a reduced number of pilot symbols. The approach leverages score-based generative models to perform posterior sampling, and represents a new research direction for MIMO channel estimation. Our results on simulated mmWave channels show that performance is favourable in-distribution, as well as in out-of-distribution settings, where channel distributions not seen during training are tested. Additionally, compared to prior work, estimation with posterior sampling has the advantage of avoiding error floors \revision{and achieving high-quality estimation in the high SNR regime, which leads to improved end-to-end performance, while trading off complexity.}

A current limitation of score-based models is the high inference complexity of \revision{posterior sampling with Langevin dynamics}. In Section~\ref{sec:ablation} we have performed an initial investigation into improvements achievable through architectural modifications (depth and width of the deep score-based model) which was able to reduce inference latency down to $1.5$ seconds for channels of size $16 \times 64$. This is usable in scenarios with low mobility, where long-term distributional shifts may still occur, \revision{such as fixed access mmWave, reflective intelligent surfaces, or mmWave backhaul. Based on the fast convergence of training in ten epochs when learning from randomly initialized models, future work that considers adapting pre-trained models to new environments using a limited amount of training channels is a promising direction}. Recent work \cite{salimans2022progressive} has investigated algorithmic improvements to posterior sampling by choosing the hyper-parameters during inference to trade-off performance for latency. This a promising direction of future research that we aim to explore, \revision{alongside others such as leveraging pretrained score-based models for channel estimation in few-bit quantization receivers and overcoming error floors in interference scenarios.}
\section*{Appendix}
\begin{theorem}[Theorem 1.1 from \cite{jalal2021instance}]
\label{theorem:ajil}
Let $h_1$ be a high-dimensional arbitrary distribution over an $\ell_2$ ball of radius $r$, and let $\mathbf{h}^\star$ be drawn from $h_1$. Let $h_2$ be any distribution on the same probability space as $h_1$, such that
\begin{equation}
\mathcal{W}_2 ( h_1, h_2 ) \le \sigma \delta^{1/2}.
\label{eq:theorem_ineq}
\end{equation}
Suppose there exists an algorithm that recovers an estimate $\hat{\mathbf{h}}$ of $\mathbf{h}^\star$ using an arbitrary measurement matrix $\mathbf{A}$ that gives $m$ measurements, and noise power $\sigma^2$, such that $\norm{\mathbf{h}^\star - \hat{\mathbf{h}}}_2 \lesssim \sigma$ with probability $1 - O(\delta)$. Then, posterior sampling with respect to $h_2$ and with $ m' \ge O ( m \log ( 1 + \frac{m r^2 \norm{\mathbf{A}}_{\infty}^2}{\sigma^2} ) + \log 1/\delta )$ noisy, Gaussian measurements at noise level $\sigma$ will output $\mathbf{h}_\mathrm{est}$ such that:
\begin{equation}
    \norm{\mathbf{h}^\star - \mathbf{h}_\mathrm{est}}_2 \lesssim \sigma \quad \mathrm{w.p.} \quad 1 - O(\delta).
\end{equation}
\end{theorem}
\begin{proof}
We refer readers to \cite{jalal2021instance} for a complete proof, omitted here due to lack of space.
\end{proof}

From \eqref{eq:theorem_ineq}, there exists a unique value $\delta_\mathrm{min}$ that is a function of $h_1$ and $h_2$, such that:
\begin{equation}
    \delta_\mathrm{min}^2 ( h_1, h_2 ) = \frac{\mathcal{W}_2^2 ( h_1, h_2 )}{\sigma_\mathrm{pilot}^2}.
\end{equation}
Let $\delta_\mathrm{MNR}(h_1, h_2) = \delta_\mathrm{min}(h_1, h_2)$. We have the following lemmas.

\begin{lemma}[Section 1.2, Page 4 from \cite{panaretos2019statistical}]
\label{lemma:bound}
Let $x$ and $y$ be two random variables with marginal distributions $p_x$ and $p_y$, respectively, and arbitrary joint distribution. Let $a$ and $b$ be two independent random variables with the same marginals. We have that:
$$
\mathcal{W}_2^2(x, y) \le \mathcal{W}_2^2(a, b).
$$
\end{lemma}

\begin{lemma}[Section 1.2, Page 7 from \cite{panaretos2019statistical}]
\label{lemma:w2_scalar}
For any two scalar random variables $a$ and $b$ with cumulative distribution functions $F_a$ and $F_b$ respectively, we have that:
$$
\mathcal{W}^2_2(a, b) = \int_{0}^{1} | F_a^{-1}(u) - F_b^{-1}(u) |^2 du.
$$
\end{lemma}

\begin{lemma}[Lemma 1 from \cite{mordant2022measuring}]
\label{lemma:w2_separable}
For any two random vectors $x$ and $y$ with mutually independent entries, we have that:
$$
\mathcal{W}^2_2(x, y) = \sum_i \mathcal{W}_2^2(x_i, y_i).
$$
\end{lemma}
\noindent Using Lemma~\ref{lemma:w2_scalar} we compute the closed-form 2-Wasserstein distance between any two scalar random variables. This leads to the following corollaries.
\begin{corollary}
\label{corollary:one}
Let $h_1 \sim \mathcal{CN}(0, \sigma_1^2)$ and $h_2 \sim \mathcal{CN}(0, \sigma_2^2)$ be two independent, complex, zero-mean and circularly symmetric Gaussian random variables. We have that:
$$
\mathcal{W}_2^2(h_1, h_2) = \left( \sigma_1 - \sigma_2 \right)^2.
$$
\end{corollary}
\begin{proof}
Let $h_i=x_i+jy_i$ for $i=1,2$. Using Lemma~\ref{lemma:w2_scalar} we have that:
\begin{align*}
\mathcal{W}_2^2 \left( h_1, h_2 \right) & = \inf_{\psi_1 \sim h_1, \psi_2 \sim h_2} \mathbb{E}\left( \norm{ \psi_1 - \psi_2 }_2^2\right) \\
 & = \inf_{\psi_1 \sim h_1, \psi_2 \sim h_2} \mathbb{E}\left( \norm{ \textrm{Re}\{\psi_1\} - \textrm{Re}\{\psi_2\} }_2^2\right) \\ & \qquad + \mathbb{E}\left( \norm{ \textrm{Im}\{\psi_1\} - \textrm{Im}\{\psi_2\} }_2^2\right) \\
 & = \mathcal{W}_2^2 \left( x_1, x_2 \right) + 
 \mathcal{W}_2^2 \left( y_1, y_2 \right) \\ & = \frac{(\sigma_1 - \sigma_2) ^ 2}{2} + \frac{(\sigma_1 - \sigma_2) ^ 2}{2} = (\sigma_1 - \sigma_2) ^ 2. \qed
\end{align*}
\renewcommand{\qedsymbol}{}
\end{proof}
\begin{corollary}
\label{corollary:two}
Let $\tau_1 = -\alpha_1 \log X_1$ and $\tau_2 = -\alpha_2 \log X_2$, where $X_1$ and $X_2$ are i.i.d. and distributed as $\mathcal{U}(0, 1)$. Then:
$$
\mathcal{W}_2^2(\tau_1, \tau_2) = 2 \left( \alpha_1 - \alpha_2 \right)^2.
$$
\end{corollary}
\begin{proof}
Using the definition of the cumulative distribution function, we have:
\begin{align*}
F_{\tau}(u) & = \mathrm{P}\left( \tau \le u \right) = \mathrm{P} \left( -\alpha \log X \le u \right) \\
& = \mathrm{P} \left( \log X \ge - \frac{u}{\alpha} \right) = \mathrm{P} \left( X \ge \exp \left( - \frac{u}{\alpha} \right) \right) \\ & = 1 - \exp \left( -\frac{u}{\alpha} \right),
\end{align*}
\noindent where we use that $X \sim \mathcal{U}(0, 1)$ satisfies $F_X(u) = u$ for all $u \in [0, 1]$. We have that $F^{-1}_{\tau}(v) = -\alpha \log \left( 1 - v \right)$. Using Lemma~\ref{lemma:w2_scalar}:
\begin{align*}
    \mathcal{W}_2^2 \left( \tau_1, \tau_2 \right) & = \int_0^1 | F^{-1}_{\tau_1}(v) - F^{-1}_{\tau_2}(v) |^2 dv \\ & = \int_0^1 | -\alpha_1 \log \left( 1 - v \right) +  \alpha_2 \log \left( 1 - v \right)|^2 dv \\
    & = \left( \alpha_1 - \alpha_2 \right)^2 \int_0^1 \log^2 \left( 1 - v \right) dv \\ & = \left( \alpha_1 - \alpha_2 \right)^2 \int_0^1 \log^2 \left( t \right) dt,
\end{align*}
\noindent where the last equality is by change of variable $t = 1 - v$. Using that $\int_0^1 \log^2 ( t ) dt = ( t ( \log |t| )^2 - 2 t \log |t| + 2t ) \rvert_0^1 = 2$, the proof is completed.
\end{proof}
\noindent Combining Corollary~\ref{corollary:one}, Corollary~\ref{corollary:two}, and Lemma~\ref{lemma:w2_separable}, we obtain that:
$$
\delta_\mathrm{MNR}(a, b) = \sum_{i=1}^K \left( (\sigma_i^{(1)} - \sigma_i^{(2)})^2 + 2 (\alpha_i^{(1)} - \alpha_i^{(2)})^2 \right) / \sigma_\mathrm{pilot^2},
$$ for two channel distributions $a$ and $b$ that satisfy the assumptions as $h_1$ and $h_2$ in Section~\ref{sec:theory}, as well as an additional mutual independence between all tap location and delays. Using Theorem~\ref{theorem:ajil} and Lemma~\ref{lemma:bound}, we obtain that $\delta_\mathrm{MNR}(h_1, h_2) \le \delta_\mathrm{MNR}(a, b)$ and Theorem~\ref{theorem:two} is proved.

\bibliographystyle{IEEEtran}
\bibliography{mybib}

\end{document}